%% file: paper.tex
\documentclass[sigconf]{acmart}

\renewcommand{\textrightarrow}{$\rightarrow$}

\usepackage{marginnote}
\usepackage{booktabs}
\usepackage{amsmath}
\usepackage{amssymb}
\usepackage{amsfonts}
\usepackage{amsthm}
\usepackage{mathtools}
\def\BibTeX{{\rm B\kern-.05em{\sc i\kern-.025em b}\kern-.08em
    T\kern-.1667em\lower.7ex\hbox{E}\kern-.125emX}}
\usepackage{algorithm}
\usepackage[noend]{algpseudocode}
\usepackage{thmtools} 
\usepackage{thm-restate}
\usepackage{multirow}
\usepackage{lipsum}
\usepackage[scaled=.8]{helvet}
\usepackage[scaled=.8]{beramono}
\usepackage{float}
\usepackage{graphicx}
\usepackage{subcaption}
\usepackage{adjustbox}
\usepackage{fancyvrb}
\usepackage{url}
\usepackage{hyperref}
\usepackage{xspace}
\usepackage{paralist}
\usepackage{thm-restate}
\usepackage{tikz}
\usepackage{chngcntr}
\usetikzlibrary{shapes,positioning,calc}
\colorlet{lightgray}{gray!20}
\input{macros}

\setcopyright{none}

\acmDOI{}

\acmISBN{}

\acmConference{}{}{}
\acmYear{2019}
\copyrightyear{}

\acmPrice{}

\acmSubmissionID{}

\begin{document}

\title{\bf Repairing mappings under policy views}

\author{Angela Bonifati}
\affiliation{%
  \institution{Lyon 1 University \& Liris CNRS}
  \city{Lyon, France}
}
\email{angela.bonifati@univ-lyon1.fr}

\author{Ugo Comignani}
\affiliation{%
  \institution{Lyon 1 University \& Liris CNRS}
  \city{Lyon, France}
}
\email{ugo.comignani@univ-lyon1.fr}

\author{Efthymia Tsamoura}
\affiliation{%
  \institution{Alan Turing Institute \&\\ University of Oxford}
  \city{Oxford, UK}
}
\email{efthymia.tsamoura@cs.ox.ac.uk}

\renewcommand{\shortauthors}{A. Bonifati, U. Comignani, and E. Tsamoura}

\input{abstract}



\maketitle

\input{intro}
\input{related}
\input{prelims}
\input{privacy}

\input{repair}

\input{experiments}

\input{conclusion}

\bibliographystyle{ACM-Reference-Format}
\bibliography{references}

%
%

\end{document}

%% file: abstract.tex
\begin{abstract}    
The problem of data exchange involves a source schema, a target schema
and a set of mappings from transforming the data between the two schemas. 
We study the problem of data exchange 
in the presence of privacy restrictions on the source.
The privacy restrictions are expressed as a set of 
\emph{policy views} representing the information 
that is safe to expose over \emph{all} instances of the source. 
We propose a protocol that provides formal privacy guarantees
and is \emph{data-independent}, i.e., if certain criteria are met, 
then the protocol guarantees that the mappings leak no sensitive information
independently of the data that lies in the source. 
We also propose an algorithm for \emph{repairing} 
an input mapping w.r.t. a set of policy views, 
in cases where the input mapping leaks sensitive information.
The empirical evaluation of our work shows that the proposed algorithm is quite efficient,
repairing sets of 300 s-t tgds in an average time of 5s on a commodity machine.
To the best of our knowledge, our work is the first one  
that studies the problems of exchanging data and repairing mappings 
under such privacy restrictions. 
Furthermore, our work is the first to provide practical algorithms
for a logical privacy-preservation paradigm, described 
as an open research challenge in previous work on this area.
\end{abstract}

%% file: intro.tex
\section{Introduction}

We consider the problem of exchanging data between a source schema 
$\source$ and a target schema $\target$ via 
a set of \emph{source-to-target} (s-t) dependencies $\Sigma_{st}$ 
that usually come in the form of \emph{tuple-generating dependencies} (tgds).
This triple of a source schema, a target schema 
and a set of dependencies is called a \emph{mapping}.
The s-t dependencies specify how and what source data should appear 
in the target and are expressed as sentences in first-order logic \cite{Fagin2005a}. 

Our work considers a privacy-aware variant of the data exchange problem, 
in which the source comes with a set of constraints, representing 
the data that is \emph{safe} to expose to the target over \emph{all instances} of the source.   
We also assume that all users, both the malicious and the non-malicious ones, 
might know the source and the target schema, the data in the target as well as the s-t tgds. 
Under these assumptions, our work will address the following issues: 
how could we represent privacy restrictions on the sources and what would it mean 
for a data exchange setting to be safe under the proposed privacy restrictions?; 
assuming that the privacy-preservation protocol is fixed, how could we assess the safety of 
a data exchange setting w.r.t. the privacy restrictions and provide \emph{strong guarantees} of no privacy leak?;
finally, in case of privacy violations, how could we \emph{repair} the s-t tgds?    

Regarding the first issue, we assume that the restrictions on the sources 
are expressed as a set of views, called \emph{policy views}.  
Inspired by prior work on privacy-preservation \cite{NashD07,MCCK17},
we define a set of s-t tgds to be safe w.r.t. the policy views
if \emph{every positive information} that is kept secret
by the policy views is also kept secret by the s-t tgds.
As we will see in subsequent sections, the proposed privacy-preservation 
protocol is \emph{data-independent} allowing us to provide strong privacy-preservation guarantees 
over all instances of the sources. The above addresses the second
aforementioned issue, as well. 
Regarding the third issue, our work proposes a 
repairing algorithm for the proposed privacy-preservation protocol.   
The feature of the proposed repairing algorithm is that it can employ techniques for 
learning the user preferences during the repairing process. 
The empirical evaluation of our work over an existing benchmark shows that the 
proposed algorithm is quite efficient. Indeed it can 
repair a set of 300 s-t tgds in less than 5s on a commodity machine. 

Our secure data exchange setting is exemplified in Figure
\ref{fig:scenario} and 
in the following running example inspired by a real world scenario from an hospital in the UK\footnote{\url{https://www.nhs.uk/}}.
\input{fig_scenario}

\begin{example}\label{example:running:intro} 
    Consider the source schema $\source$ consisting of the following relations: $\pat$, $\hospn$, $\hosps$, $\onc$ and $\school$.
    Relation $\pat$ stores for each person registered with the NHS, his insurance number, his name, his ethnicity group and his county.
    Relations $\hospn$ and $\hosps$ store for each patient who has been admitted to 
    some hospital in the north or the south of UK, his insurance number and the reason 
    for being admitted to the hospital. Relation $\onc$ stores information related 
    to patients in oncology departments and, in particular, their 
    insurance numbers, their treatment and their progress. Finally, relation $\school$ stores for each student 
    in UK, his insurance number, his name, his ethnicity group and his county.      
    
    Consider also the set $\policyViews$ comprising the policy views $\viewPred_1$--$\viewPred_4$.
    The policy views define the information that is safe to make available to public. 
    View $\viewPred_1$ projects the ethnicity groups and the hospital admittance reasons 
    for patients in the north of UK;
    $\viewPred_2$ projects the counties and the hospital admittance reasons 
    for patients in the north of UK;
    $\viewPred_3$ projects the treatments and the progress of patients of oncology departments;
    $\viewPred_4$ projects the ethnicity groups of the school students.  
    The policy views are safe w.r.t. the NSS privacy preservation protocol. Indeed,
    the NSS privacy preservation protocol considers as unsafe any non-evident piece of information that 
    can potentially de-anonymize an individual.
    For example, views $\viewPred_1$ and $\viewPred_2$
    do not leak any sensitive information, since the results concern patients 
    from a very large geographical area and, thus, the probability of de-anonymizing a patient is 
    significantly small. For similar reasons, views $\viewPred_3$ and $\viewPred_4$ are considered to be safe: 
    the probability of de-anonymizing patients of the oncology department 
    from $\viewPred_3$ is zero, since there is no way to link
    a patient to his treatment or his progress, while $\viewPred_4$ projects information which is already evident to public.
    \begin{align}
        \pat(\pid,\name,\ethn,\post) \wedge \hospn(\pid,\dis)  &\leftrightarrow \viewPred_1(\ethn,\dis)      \label{example:pat-hosp1}\\
        \pat(\pid,\name,\ethn,\post) \wedge \hosps(\pid,\dis)  &\leftrightarrow \viewPred_2(\post,\dis)      \label{example:pat-hosp2}\\
        \onc(\pid,\tr,\pr)                                     &\leftrightarrow \viewPred_3(\tr,\pr)         \label{example:onc}\\
        \school(\sid,\name,\ethn,\post)                        &\leftrightarrow \viewPred_4(\ethn)           \label{example:sch}
    \end{align}    
    
    Finally, consider the following set of s-t dependencies $\Sigma_{st}$.
    The dependencies $\mu_e$ and $\mu_c$ project similar information
    with the views $\viewPred_1$ and $\viewPred_2$, respectively.
    They focus, however, on patients in the north of UK.
    Finally, the dependency $\mu_s$ projects the ethnicity groups of 
    students who have been in some oncology department.   
    \begin{align}
        \pat(\pid,\name,\ethn,\post) \wedge \hospn(\pid,\dis)     &\rightarrow \ethreason(\ethn,\dis)     \tag{$\mu_e$} \nonumber \\ 
        \pat(\pid,\name,\ethn,\post) \wedge \hospn(\pid,\dis)     &\rightarrow \countyreason(\post,\dis)  \tag{$\mu_c$} \nonumber \\ 
        \school(\sid,\name,\ethn,\post) \wedge \onc(\sid,\tr,\pr) &\rightarrow \studentOnc(\ethn)         \tag{$\mu_s$} \nonumber
    \end{align}
    The questions addressed in our paper are the following ones:
    Are the s-t dependencies safe w.r.t the policy views?
    Are there any formal guarantees for privacy preservation in the context of policy views?
    If the s-t dependencies are not safe w.r.t. the policy views,
    how could we repair them and provide formal privacy preservation guarantees?
\end{example}
 
Our technique is inherently data-independent thus bringing the advantage that
both the safety test and the repairing operations are executed on the
metadata provided through the mappings and not on the underlying data
instances. The logical foundations of information
disclosure in ontology-based data integration have been laid in
\cite{MCCK17} in the presence of boolean policies. 
Instead, we focus on non-boolean policies. 
Taking a step forward, we also propose an algorithm
for repairing a set of unsafe s-t tgds w.r.t. 
our privacy preservation protocol. 
To the best of our knowledge, our work is the first to provide practical algorithms
for a logical privacy-preservation paradigm, described as an open research challenge in \cite{NashD07,MCCK17}.
We leave out probabilistic approaches and anonymization techniques
\cite{MiklauS07,Sweene02}, which involve modifications 
of the underlying data instances and are orthogonal
to our approach. A careful treatment of related work 
is deferred to Section \ref{sec:related}. 

The source code and the experimental scenarios are publicly available
at \url{https://github.com/ucomignani/MapRepair.git}.

The paper is organized as follows. Section~\ref{sec:related} discusses the
related work. Section~\ref{sec:prelims} presents the
basic concepts and notions. Section~\ref{sec:privacy} lays our
privacy preservation protocol. Section~\ref{sec:repair}
presents our repairing algorithms and their properties. 
mechanism.
Section~\ref{sec:experiments} outlines the experimental results,
while Section~\ref{sec:conclusion} concludes our paper. 

%% file: fig_scenario.tex
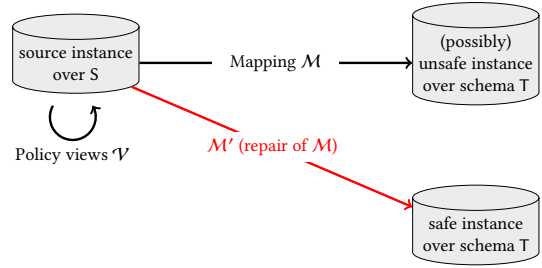
\begin{figure}
  \centering
  \adjustbox{width=\columnwidth, height=3.5cm ,keepaspectratio}{
    \begin{tikzpicture}[
      database/.style={cylinder,
      cylinder uses custom fill,
      cylinder body fill=gray!15,
      cylinder end fill=gray!15,
      shape border rotate=90,
      aspect=0.2,
      draw
      },
      view/.style={rectangle,draw,fill=gray!15}
    ]

    \node[database] (sourceDb) at (0,3) {\begin{minipage}{2cm}\centering source instance\\ over $\source$ \end{minipage}};
    \node[database] (unsafeDb) at (7,3) {\begin{minipage}{2cm}\centering (possibly) unsafe instance over schema $\target$ \end{minipage}};
    \node[database] (safeDb) at (7,0) {\begin{minipage}{2cm}\centering safe instance over schema $\target$ \end{minipage}};

    \draw[thick,->,,line width=1.2pt] (-0.4,2.25) arc (150:150+240:4mm)  node[pos=0.5,below]{\begin{minipage}{2cm}\centering Policy views $\policyViews$ \end{minipage}} (sourceDb);
    
    \draw[->,line width=1.2pt](sourceDb) -- (unsafeDb) node [midway,fill=white] {\begin{minipage}{2cm}\centering Mapping $\mappingsOriginal$\end{minipage}};

    \draw[red,->,line width=1.2pt](sourceDb) -- (safeDb) node [midway,fill=white] {
    \begin{minipage}{2.3cm}\centering $\mappingsOriginal^\prime$ (repair of $\mappingsOriginal$)\end{minipage}};
	
    \end{tikzpicture}
   }
 \caption{A data exchange setting with mappings and policy views.}
 \label{fig:scenario}
\end{figure}

%% file: related.tex
\section{Related Work} \label{sec:related}

{\bf Privacy in data integration} Safety of secret queries formulated against a global schema and adhering to 
the certain answers semantics has been tackled in previous theoretical work \cite{NashD07}. 
They define the optimal attack that characterizes a set of queries that an attacker can issue to which no further queries can be added to infer more information.
They then define the privacy guarantees against the optimal attack by considering the static and the dynamic case, the latter 
corresponding to modifications of the schemas or the GLAV mappings. 
The same definition of secret queries and privacy setting is adopted in \cite{MCCK17}, which instead focuses on boolean conjunctive 
queries as policy views and on the notion of safety with respect to a given mapping. An ontology-based integration scenario is assumed in which the target instance is produced via a set of mappings starting from an underlying data source. Whereas they study the complexity of the view compliance problem in both data-dependent and data-independent setting, we focus on the latter and extend it to non-boolean conjunctive queries as policy views. 
We further consider multiple policy views altogether in the design practical algorithm for checking the safety of schema
mappings and for repairing the mappings in order to resume safety in case of violations. 

{\bf Privacy in data publishing}
Data publishing accounts for the settings in which a view exports or publishes the information of an underlying data source. 
Privacy and information disclosure in data publishing linger over the problem of avoiding the disclosure of the content of the view under a confidential query. 
A probabilistic formal analysis of the query-view security model has been presented in \cite{MiklauS07}, 
where they offer a complete treatment of the multi-party collusion and the use of external adversarial knowledge. 
Access control policies using cryptography are used in \cite{MiklauS07} to enforce the 
authorization to an XML document. Our work differs from theirs on both the considered setting, as well as the adopted techniques 
and the adopted privacy protocol.

{\bf Controlled Query Evaluation}
Controlled Query Evaluation is a confidentiality
enforcement framework introduced in \cite{sicherman1983answering} and refined in \cite{bonatti1995foundations},\cite{biskup_controlled_2004} and \cite{biskup2008keeping}, in which a policy declaratively specifies sensitive information
and confidentiality is enforced by a censor. Provided a query as input, a censor verifies whether the query leads to a violation of the policy and in case of a violation it returns a distorted
answer. It has been recently adopted in ontologies expressed with Datalog-like rules and in lightweight Description Logics \cite{grau_controlled_2015}. 
They assume that the policies are only known to database administrators and not to ordinary users and that the data has protected access through a query interface. Our assumptions and setting are quite different, since our multiple policy views are accessible to every user and our goal is to render the s-t mappings safe with respect to a set of policies via repairing and rewriting.  

{\bf Data privacy}
Previous work has addressed access control to protect database instances at
different levels of granularity \cite{SarfrazNCB15}, in
order to combine encrypted query processing and authorization rules. Our work does
not deal with these authorization methods, as well as does
not consider any concrete privacy or anonymization algorithms operating on
data instances, such as differential privacy \cite{DworkR14} 
and k-anonymity \cite{Sweene02}. 

%% file: prelims.tex
\section{Preliminaries} \label{sec:prelims}

Let $\const$, $\nulls$, and $\vars$ be mutually disjoint, infinite sets of
\emph{constant values}, \emph{labeled nulls}, and \emph{variables},
respectively. 
A \emph{schema} is a set of \emph{relation names} (or just \emph{relations}),
each associated with a nonnegative integer called \emph{arity}. 
A \emph{relational atom} has the form ${R(\vec t)}$ where $R$ is an
$n$-ary relation and ${\vec t}$ is an $n$-tuple of \emph{terms}, 
where a term is either a constant, a labelled null, or a variable. 
An \emph{equality atom} has the form ${t_1 = t_2}$ where $t_1$ and $t_2$ are terms.
An atom is called ground or fact, when it does not contain any variables.
A position in an $n$-ary atom $A$ is an integer ${1 \leq i \leq n}$.
We denote by $\pos{A}{i}$, the $i$-th term of $A$.  
An instance $I$ is a set of relational facts.
An atom (resp.\ an instance) is \emph{null-free} if it does not contain labelled nulls.
The \emph{critical instance} of a schema $\source$, denoted as $\critical{\source}$,
is the instance containing a fact of the form ${R(\vec{\ast})}$, 
for each $n$-ary relation ${R \in \source}$, where $\ast$ is called the \emph{critical constant}
and $\vec{\ast}$ is an $n$-ary vector.  
A \emph{substitution} $\sigma$ is a mapping from variables into constants or labelled nulls.

A \emph{dependency} describes the semantic relationship between relations. 
A \emph{Tuple Generating Dependency} (tgd) is a formula of the form 
${\forall \vec x ~ \lambda(\vec x) \rightarrow \exists \vec y ~ \rho(\vec x, \vec y)}$, where
$\lambda(\vec x)$ and $\rho(\vec x,\vec y)$ are conjunctions of relational, null-free
atoms. An \emph{Equality Generating Dependency} (egd) is a formula of the
form ${\forall \vec x ~ \lambda(\vec x) \rightarrow x_i = x_j}$, 
where $\lambda(\vec x)$ is a conjunction of relational, null-free atoms.
We usually omit the quantification for brevity. 
We refer to the left-hand side of a tgd or an egd $\delta$ as the \emph{body}, denoted as $\body(\delta)$, 
and to the right-hand side as the \emph{head}, denoted as $\head(\delta)$. 
An instance $I$ satisfies a dependency $\delta$, written ${I \models \delta}$ if 
each homomorphism from $\body(\delta)$ into $I$ can be extended to a homomorphism 
$h'$ from $\head(\delta)$ into $I$. An instance $I$ satisfies a set of dependencies 
$\Sigma$, written as ${I \models \Sigma}$, if ${I \models \delta}$ holds, for each ${\delta \in \Sigma}$.  
The \emph{solutions} of an instance $I$ w.r.t. $\Sigma$ is the
set of all instances $J$ such that ${J \supseteq I}$ and ${J \models \Sigma}$. 
A solution is called \emph{universal} if it can be homomorphically embedded to each solution of 
$I$ w.r.t. $\Sigma$.

A \emph{conjunctive query} (CQ) is a formula of the form ${\exists \vec y ~
\bigwedge_i A_i}$, where $A_i$ are relational, null-free atoms.
A CQ is boolean if it does not contain any free variables. 
A substitution $\sigma$ is an \emph{answer} to a CQ $\util$ on an instance $I$ if the domain of $\sigma$
is the free variables of $Q$, and if $\sigma$ can be extended to a
homomorphism from ${\bigwedge_i A_i}$ into $I$. 
We denote by $\util(I)$, the answers to $\util$ on $I$. 

Given an instance $I$ and a set of dependencies $\Sigma$,
the chase iteratively computes a \emph{universal} solution of $I$ w.r.t. $\Sigma$ 
\cite{Fagin2005a, chasebench}.
Starting from ${I_0 = I}$, at each iteration $i$, 
it computes a new instance $I_i$ by applying a tgd or an egd chase step: 

    \myparagraph{tgd chase step} Consider an instance $I_i$, 
    a tgd $\delta$ of the form ${\forall \vec x ~ \lambda(\vec x) \rightarrow \exists \vec y ~ \rho(\vec x, \vec y)}$
    and a homomorphism $h$ from $\lambda(\vec x)$ into $I_i$, such that there does not exist 
    an extension of $h$ to a homomorphism from $\rho(\vec x, \vec y)$ into $I_i$.
    Such homomorphisms are called \emph{active triggers}.  
    Applying the tgd chase step for $\delta$ and $h$ to $I_i$ results in a new instance  
    ${I_{i+1} = I_i \cup h'(\rho(\vec x, \vec y))}$, where $h'$ is a substitution such that ${h'(x_j) = h(x_j)}$ for
    each variable ${x_j \in \vec x}$, and $h'(y_j)$, for each ${y_j \in \vec y}$,
    is a fresh labeled null not occurring in $I_i$. 

    \myparagraph{egd chase step} Consider an instance $I_i$, 
    an egd $\delta$ of the form ${\forall \vec x ~ \lambda(\vec x) \rightarrow x_i = x_j}$
    and a homomorphism $h$ from $\lambda(\vec x)$ into $I_i$, such that ${h(x_i) \neq h(x_j)}$. 
    Applying the egd chase step for $\delta$ and $h$ to $I_i$ \emph{fails} if 
    ${\{ h(x_i),h(x_j)\} \subseteq \const}$, and otherwise it
    results in a new instance ${I_{i+1} = \nu(I_i)}$, 
    where ${\nu = \{ h(x_j) \mapsto h(x_i) \}}$ if ${h(x_i) \in
    \const}$, and ${\nu = \{ h(x_i) \mapsto h(x_j) \}}$ if ${h(x_i) \not\in \const}$.
We denote by $\chase{I}{\Sigma}$, the chase of $I$ w.r.t. $\Sigma$. 

Let $\source$ be a source schema and let $\target$ be a target schema. 
A \emph{mapping} $\mappingsOriginal$ from $\source$ to $\target$ 
is defined as a triple $(\source,\target,\Sigma)$, where, generally,  
$\Sigma = \Sigma_s \cup \Sigma_{st} \cup \Sigma_t$.   
$\Sigma_s$, $\Sigma_{st}$ and $\Sigma_t$ denote the source, 
s-t and target dependencies over $\source$ and $\target$, 
respectively. We usually refer to the dependencies in 
$\Sigma_{st}$ as \emph{mappings}. A variable $x$ of a mapping ${\mu \in \Sigma_{st}}$
is called exported if it occurs both in the body and the head of $\mu$.
We denote by $\frontier{\mu}$, the set of exported variables of $\mu$.  
The inverse of set of s-t dependencies $\Sigma_{st}$, 
denoted as $\inverse{\Sigma_{st}}$ is the set consisting, for each 
mapping $\mu$ in $\Sigma_{st}$ of the form ${\lambda(\vec x) \rightarrow \rho(\vec x, \vec y)}$, 
a mapping $\inverse{\mu}$ of the form ${\rho(\vec x, \vec y) \rightarrow \lambda(\vec x)}$. 
In this paper, we will focus on the scenario,
where $\Sigma_s$ and $\Sigma_t$ are empty, so $\Sigma$ will 
be equal to $\Sigma_{st}$. Furthermore, in this paper we focus on \emph{GLAV} mappings, i.e., 
s-t dependencies corresponding to a set of views. 

The \emph{certain answers} of a CQ $\util$ over $\target$ w.r.t 
$\instance$ and $\mappingsOriginal$, denoted as 
$\certain{\util}{\instance}{\mappingsOriginal}$, 
are the intersection of all answers to $\util$ over all 
solutions of $\instance$ w.r.t. $\Sigma$. 
Given a finite, null-free instance $I$ of the source schema, the objective of \emph{data
exchange} is to compute a universal solution 
of $I$ w.r.t. the dependencies $\Sigma$ from $\mappingsOriginal$.

%% file: privacy.tex
\section{Privacy preservation} \label{sec:privacy}

In this section, we introduce our notion of privacy preservation.  
Let $\policyViews$ be a set of \emph{policy views} over $\source$.
The policy views represent the information that is safe to expose for instances $\instance$ of $\source$.
We denote by $\mappingsOriginal_{\viewSchema} = (\source,\viewSchema,\policyViews)$ 
the mapping from $\source$ to $\viewSchema$, where
$\viewSchema$ denotes the schema of the views occurring in $\policyViews$.    
Our goal is to verify whether a user-defined mapping $\mappingsOriginal = (\source,\target,\Sigma)$ 
is safe w.r.t. a view mapping $\mappingsOriginal_{\viewSchema}$.   
Below, we will introduce a notion for assessing the safety of a GAV mapping $\mappingsOriginal_2$ 
with respect to a GAV mapping $\mappingsOriginal_1$, when both make use of the same source schema $\source$. 
Below, let $\Sigma_i = \Sigma_{st_i}$ be the dependencies associated with $\mappingsOriginal_i$.  

\subsection{A formal privacy-preservation protocol}\label{sec:privacy:protocol}
Our notion of privacy preservation builds upon the protocol introduced in \cite{MCCK17}.
Below, we formalize the notion of privacy preservation from \cite{MCCK17} 
and we extend it for non-boolean conjunctive queries.    
First we recapitulate the notion of indistinguishability of two source instances.  

\begin{definition}\label{definition:indistinguishability}
    Two instances $\instance$ and $\instance'$ of a source schema $\source$ 
    are indistinguishable with respect to a mapping ${\mappingsOriginal = (\source,\target,\Sigma)}$, 
    denoted as $\indistinguishability{\instance}{\instance'}{\mappingsOriginal}$,
    if ${\certain{\util}{\instance}{\mappingsOriginal} = \certain{\util}{\instance'}{\mappingsOriginal}}$
    for each CQ $\util$ over $\target$.
\end{definition}

Informally, Definition~\ref{definition:indistinguishability} tells us that two source instances are 
indistinguishable from each other if the target instances have the same certain answers. 

\begin{definition}\label{definition:disclosure}
    A mapping ${\mappingsOriginal = (\source,\target,\Sigma)}$ 
    \emph{does not disclose a CQ $\secret$ over $\source$ on any instance of $\source$}, if 
    for each instance $\instance$ of $\source$ there exists an instance $\instance'$ such that 
    ${\indistinguishability{\instance}{\instance'}{\mappingsOriginal}}$ and ${\secret(\instance') = \emptyset}$. 
\end{definition}

The problem of checking whether 
a mapping $\mappingsOriginal$ over $\source$ does not disclose a boolean 
and constants-free CQ $\secret$ on any instance of $\source$ is decidable for GAV mappings 
consisting of CQ views \cite{MCCK17}. 
In particular, $\mappingsOriginal$ does not disclose $\secret$ 
on any instance of $\source$ if and only if there does not exist a homomorphism 
from $\secret$ into the \emph{unique} instance computed by the 
\emph{visible chase} $\vchase{\Sigma}{\source}$ of $\Sigma$ under the critical instance $\critical{\source}$ of $\source$.
The visible chase computes a \emph{universal source instance}-- that is an instance, 
such that the visible part of any instance of $\source$  (i.e., the subinstance that becomes available through the mappings) can be mapped into it.
The only constant occurring in the instance computed by $\vchase{\Sigma}{\source}$ is the critical constant $\ast$
and it represents any other constant that can occur in the source instance. 
 
For the purpose of repairing the mappings efficiently, we introduce our own variant of the visible chase, 
which organizes the facts derived during chasing into subinstances
called \emph{bags}. Algorithm~\ref{alg:bagchase} describes the steps of the proposed variant. 
Please note that Algorithm~\ref{alg:bagchase} derives the same set of facts with the algorithm from \cite{MCCK17}.
However, instead of keeping these facts in a single set, we keep them in separate bags. 
Before presenting Algorithm~\ref{alg:bagchase}, we will introduce some new notions.

\begin{definition}\label{definition:derived-egd}
	Consider an instance $I$. Consider also a s-t tgd $\delta$
	and a homomorphism $h$ from ${\body(\delta)}$ into $I$, such that 
	${h(x) \in \nulls}$, for some ${x \in \frontier{\delta}}$. 
	Then, we say that the egd 
	\begin{align}
		\body(\delta) \rightarrow \bigwedge \limits_{\forall x \in \frontier{\delta}:  h(x) \in \nulls} x = \ast														\label{eq:derived-egd}
	\end{align}	
	is \emph{derived} from $\delta$ in $I$.
	For an egd $\epsilon$ that is derived from a s-t tgd $\delta$ in $I$, $\origin{\epsilon}$ denotes $\delta$.
	For a set of s-t tgds $\Sigma$ and an instance $I$, $\Sigma_{\approx}$ is the set comprising for each ${\delta \in \Sigma}$, the egd that is derived from $\delta$ in $I$.
\end{definition}

\begin{definition}\label{definition:relevant-bag}
	 Consider an instance $I$, whose facts are organized into the bags ${\bag_1,\dots,\bag_m}$.
	 Consider also a derived egd $\delta$ of the form \eqref{eq:derived-egd} and an active trigger $h$ for $\delta$ in $I$. 
	 A bag $\bag_i$ is \emph{relevant} for $\delta$ and $h$ in $I$, where ${1 \leq i \leq m}$, if some fact ${F \in h(\body(\delta))}$ occurs in $\bag_i$ and 
	 if some ${h(x)}$ is a labeled null occurring in $\bag_i$, where ${x \in \frontier{\delta}}$.
	 
	 Let ${\bag_{j_1},\dots,\bag_{j_k} \subseteq \bag_1,\dots,\bag_m}$ be the set of bags 
	 that are relevant for $\delta$ and $h$ in $I$.
	 Let ${\nu = \{ h(x_j) \mapsto h(x_i) \}}$ if ${h(x_i) = \ast}$, 
	 and ${\nu = \{ h(x_i) \mapsto h(x_j) \}}$ if ${h(x_i) \not\in \const}$, where $x_i, x_j$ are variables from ${\frontier{\delta}}$.
	 Then, the \emph{derived} bag $\bag$ for $\delta$ and $h$ in $I$ consists of the facts in ${\bigcup_{l=1}^k \nu(\bag_{j_l})}$. The bags 
	 ${\bag_{j_1},\dots,\bag_{j_k}}$ are called the \emph{predecessors} of $\bag$.
	 We use ${\bag_{j_l} \prec \bag}$ to denote that $\bag_{j_l}$ is a predecessor of $\bag$, for ${1 \leq l \leq k}$. 
\end{definition}

We are now ready to proceed with the description of Algorithm~\ref{alg:bagchase}. 
Given a s-t mapping, Algorithm~\ref{alg:bagchase} computes a universal source instance whose facts are organized into bags.  
Algorithm~\ref{alg:bagchase} first computes the instance $I_0$ by 
chasing $\critical{\source}$ using the s-t tgds, line~\ref{bagchase:init}.
It then chases $I_0$ with the inverse s-t tgds $\inverse{\Sigma}$, line~\ref{bagchase:back}.
and proceeds by chasing $I_1$ with the set of all derived egds $\Sigma_{\approx}$, for each ${\delta \in \Sigma}$ in $I_1$,
line~\ref{bagchase:sigma_approx:end}. 
Algorithm~\ref{alg:bagchase} computes a fresh bag at each chase step. 
In particular, for each active trigger $h$ for $\delta$ in $I$, 
Algorithm~\ref{alg:bagchase} adds a fresh bag with facts $h'(\head(\delta))$, 
if ${\delta \in \Sigma \cup \inverse{\Sigma}}$, line~\ref{bagchase:TGDs:create};
otherwise, if ${\delta \in \Sigma_{\approx}}$, then it adds the derived bag for $\delta$ and $h$ in $I$, 
see Definition~\ref{definition:relevant-bag}, line~\ref{bagchase:EGDsChase:create}. 

Note that, $\Sigma_{\approx}$ aims at ``disambiguating'' as many labeled nulls occurring in $I_1$ 
as possible, by unifying them with the critical constant $\ast$. 
Since $\ast$ represents the information that is ``visible" to a third-party, 
chasing with $\Sigma_{\approx}$ computes the \emph{maximal information} 
from the source instance a third-party has access to.  
Note that Algorithm~\ref{alg:bagchase} always terminates \cite{MCCK17}.  
Let ${B = \vchase{\Sigma}{\source}}$.  
We will denote by $\visins{\Sigma}{\source}$, 
the instance ${\bigcup \nolimits_{\beta\in B} \beta}$.
  
\begin{algorithm}[htb]
\scriptsize
\caption{$\vchase{\Sigma}{\source}$}\label{alg:bagchase}
\begin{algorithmic}[1]  
    \State $B_0 \defeq \mathsf{bagChaseTGDs}(\Sigma,\critical{\source})$																				\label{bagchase:init}
    \State $B_1 \defeq \mathsf{bagChaseTGDs}(\inverse{\Sigma},\bigcup \nolimits_{\beta\in B_0} \beta \setminus \critical{\source})$                                   \label{bagchase:back}
    \State Let $\Sigma_\approx$ be the set of all derived egds $\Sigma_{\approx}$, for each ${\delta \in \Sigma}$ in $I_1$                                                       \label{bagchase:sigma_approx:init} 
    \State \textbf{return} $\mathsf{bagChaseEGDs}(\Sigma_\approx, B_0 \cup B_1)$																		\label{bagchase:sigma_approx:end}

    \Statex
    \Procedure{$\mathsf{bagChaseTGDs}$}{$\Sigma,I$}
        \State $B \defeq \emptyset$
        \For{\textbf{each} $\delta \in \Sigma$}                                                                           						\label{bagchase:TGDs:start}
          \For{\textbf{each} active trigger $h : \body(\delta) \rightarrow I$}                                              				\label{bagchase:TGDs:hom}
          	\State \textbf{create} a fresh bag $\beta$ with facts $h'(\head(\delta))$									\label{bagchase:TGDs:create}
            	\State \textbf{add} $\beta$ to $B$																					
          \EndFor
        \EndFor                                                                                                          									\label{bagchase:TGDs:end}
        \State \textbf{return} $B$
    \EndProcedure
    
    \Statex
    \Procedure{$\mathsf{bagChaseEGDs}$}{$\Sigma_{\approx},B$}
        \State $i \defeq 0$; $I_i \defeq \bigcup \nolimits_{\beta\in B} \beta$
        \Do
            \State ${i \defeq i + 1}$

            \For{\textbf{each} $(\delta \in \Sigma_{\approx}$ of the form \eqref{eq:derived-egd}}   						\label{bagchase:EGDsChase:init}                                                                        
                \For{\textbf{each} active trigger $h:\body(\delta) \rightarrow I_{i-1}$}									\label{bagchase:EGDsChase:homom} 
                \If{$h(x) \neq \ast$, for some ${x \in \frontier{\delta}}$}													\label{bagchase:EGDsChase:tgdsConditions} 
                	\State Let $\bag$ be the derived bag for $\delta$ and $h$ in $I_{i-1}$
                    \State \textbf{add} $\beta$ to $B$																	\label{bagchase:EGDsChase:create} 
                    \State $I_i \defeq I_i \cup \beta$
                \EndIf
                \EndFor
            \EndFor   
        \doWhile{$I_{i-1} \neq I_i$}
        \State \textbf{return} $B$
    \EndProcedure
\end{algorithmic}
\end{algorithm}

\begin{example}\label{example:running:visChase} 
    We demonstrate the visible chase algorithm over the policy views
    and the s-t dependencies from Example~\eqref{example:running:intro}.  
    
    We first present the computation of ${\visins{\policyViews}{\source} = \bigcup \nolimits_{\beta\in \vchase{\policyViews}{\source}} \beta}$.

    The critical instance $\critical{\source}$ of $\source$ consists of the facts shown in Eq.~\eqref{example:crit}.
    \begin{align}
        \pat(\ast,\ast,\ast,\ast) && \hospn(\ast,\ast) && \hosps(\ast,\ast)  \label{example:crit} \\
        \onc(\ast,\ast,\ast) && \school(\ast,\ast,\ast,\ast) &&              \nonumber
    \end{align}
    where $\ast$ is the critical constant. 
    
    The instance $I_1$ computed by chasing the output of line~\ref{bagchase:init} using $\inverse{\policyViews}$ will consist of the facts 
    \begin{align}
        \pat(\lnu_{\pid},\lnu_{\name},\ast,\lnu_{\post})   && \hospn(\lnu_{\pid},\ast)  &&   \onc(\lnu''_{\pid},\ast,\ast)                                       \tag{$I_1$} \nonumber \\
        \pat(\lnu'_{\pid},\lnu'_{\name},\lnu_{\ethn},\ast) && \hosps(\lnu'_{\pid},\ast) &&  \school(\lnu'''_{\sid},\lnu'''_{\name},\ast,\lnu'''_{\post})         \nonumber
    \end{align}
    where the constants prefixed by $\lnu$ are labeled nulls created while chasing $\critical{\source}$ with the inverse mappings. 
    Since there exists no homomorphism from the body of any s-t tgd 
    into $I_1$ mapping an exported variable into a labeled null,
    $\Sigma_\approx$ will be empty, see Definition~\ref{definition:derived-egd}.
    Thus, ${\visins{\policyViews}{\source} = I_1}$.
    
    We next present the computation of ${\visins{\Sigma_{st}}{\source} = \bigcup \nolimits_{\beta\in \vchase{\Sigma_{st}}{\source}} \beta}$.
    The instance $I'_1$ computed by chasing the output of line~\ref{bagchase:init} by $\inverse{\Sigma_{st}}$ will consist of the facts 
    \begin{align}
        \pat(\lnu_{\pid},\lnu_{\name},\ast,\lnu_{\post})   && \hospn(\lnu_{\pid},\ast)  && \school(\lnu''_{\sid},\lnu''_{\name},\ast,\lnu'_{\post})                \tag{$I'_1$} \nonumber \\
        \pat(\lnu'_{\pid},\lnu'_{\name},\lnu_{\ethn},\ast) && \hospn(\lnu'_{\pid},\ast) && \onc(\lnu''_{\sid},\lnu''_{\tr},\lnu''_{\pr})                           \nonumber
    \end{align}
    Since there exists a homomorphism from the body of $\mu_e$ into $I'_1$ mapping the exported variable $e$ into the labeled null $\lnu_{\ethn}$,
    and since there exists another homomorphism from the body of $\mu_c$ into $I'_1$ mapping the exported variable $c$ into the labeled null $\lnu_{\post}$, 
    $\Sigma_\approx$ will comprise the egds $\epsilon_1$ and $\epsilon_2$ shown below
    \begin{align}
        \pat(\pid,\name,\ethn,\post) \wedge \hospn(\pid,\dis)     &\rightarrow \ethn \approx \ast    \tag{$\epsilon_1$} \nonumber \\
        \pat(\pid,\name,\ethn,\post) \wedge \hospn(\pid,\dis)     &\rightarrow \post \approx \ast    \tag{$\epsilon_2$} \nonumber
    \end{align} 
    The last step of the visible chase involves chasing $I'_1$ 
    using $\Sigma_\approx$. WLOG, assume that the chase considers first $\epsilon_1$ and then $\epsilon_2$.
    During the first step of the chase, there exists a homomorphism
    from $\body(\epsilon_1)$ into $I'_1$.
    Hence, ${\lnu_{\ethn} = \ast}$.
    During the second step of the chase, there exists a homomorphism
    from $\body(\epsilon_2)$ into $I'_1$ and, 
    hence, ${\lnu_{\post} = \ast}$.
    The instance computed at the end of the second round of the chase will consist of the facts  
    \begin{align}
        \pat(\lnu_{\pid},\lnu_{\name},\ast,\ast)  && \hospn(\lnu_{\pid},\ast) && \hospn(\lnu'_{\pid},\ast)                 \label{example:crit-back-st-final} \\
        \school(\lnu''_{\sid},\lnu''_{\name},\ast,\lnu'_{\post}) && \onc(\lnu''_{\sid},\lnu''_{\tr},\lnu''_{\pr})          \nonumber
    \end{align}
    Since there exists no active trigger for $\epsilon_1$ or $\epsilon_2$ in 
    the instance of Eq.~\eqref{example:crit-back-st-final}, the chase will terminate.
    
     The facts in $\visins{\Sigma_{st}}{\source}$ will be organized into the following bags $\bag_1$--$\bag_5$ (one bag per line)
    \begin{align}
        \studentOnc(\ethn)        &\xrightarrow{{\langle \inverse{\mu_s}, h_1 \rangle}} \school(\lnu''_{\sid},\lnu''_{\name},\ast,\lnu'_{\post}), \onc(\lnu''_{\sid},\lnu''_{\tr},\lnu''_{\pr})            &\nonumber\\        
        \countyreason(\post,\dis) &\xrightarrow{{\langle \inverse{\mu_c}, h_2 \rangle}} \pat(\lnu'_{\pid},\lnu'_{\name},\lnu_{\ethn},\ast), \hospn(\lnu'_{\pid},\ast)                                      &\nonumber\\
        \ethreason(\ethn,\dis)    &\xrightarrow{{\langle \inverse{\mu_e}, h_3 \rangle}} \pat(\lnu_{\pid},\lnu_{\name},\ast,\lnu_{\post}), \hospn(\lnu_{\pid},\ast)                                         &\nonumber\\
        \pat(\lnu'_{\pid},\lnu'_{\name},\lnu_{\ethn},\ast), \hospn(\lnu'_{\pid},\ast) &\xrightarrow{\langle \epsilon_1, h_4 \rangle} \pat(\lnu'_{\pid},\lnu'_{\name},\ast,\ast), \hospn(\lnu'_{\pid},\ast) &\nonumber\\
        \pat(\lnu_{\pid},\lnu_{\name},\ast,\lnu_{\post}), \hospn(\lnu_{\pid},\ast)    &\xrightarrow{\langle \epsilon_2, h_5 \rangle} \pat(\lnu_{\pid},\lnu_{\name},\ast,\ast), \hospn(\lnu_{\pid},\ast)    &\nonumber
    \end{align}
    \begin{align}
        h_1 &= \{\pid \mapsto \lnu'_{\pid}, \name \mapsto \lnu'_{\name}, \ethn \mapsto \lnu_{\ethn}, \post \mapsto \ast, \dis \mapsto \ast \} \nonumber \\
        h_2 &= \{\post \mapsto \ast, \dis \mapsto \ast \} \nonumber \\
        h_3 &= \{\ethn \mapsto \ast, \dis \mapsto \ast \} \nonumber \\
        h_4 &= \{\pid \mapsto \lnu'_{\pid}, \name \mapsto \lnu'_{\name}, \ethn \mapsto \lnu_{\ethn}, \post \mapsto \ast, \dis \mapsto \ast \} \nonumber \\
        h_5 &= \{\pid \mapsto \lnu_{\pid}, \name \mapsto \lnu_{\name}, \ethn \mapsto \ast, \post \mapsto \lnu_{\post}, \dis \mapsto \ast \}   \nonumber
    \end{align}  
    The contents of the bags correspond to the right-hand side of the arrows.
    However, for presentation purposes, we also show the related dependency $\delta$ and the homomorphism $h$
    that lead to the derivation of each bag (shown at the top of each arrow), as well as, 
    the facts in $h(\body(\delta))$ (left-hand side of each arrow).
\end{example}  
   
\subsection{Preserving the privacy of policy views}\label{sec:privacy:policy_views}

We consider a 
mapping consisting of CQ views
$\mappingsOriginal = (\source,\target,\Sigma)$ to be safe w.r.t.
a view mapping consisting of CQ views $\mappingsOriginal_{\viewSchema} = (\source,\viewSchema,\policyViews)$, if $\mappingsOriginal$ 
does not disclose the information that is also not disclosed by $\mappingsOriginal_{\viewSchema}$.
Definition~\ref{definition:preservation} and Theorem~\ref{theorem:preservation_under_kernel_property} 
presented below formalize our notion of privacy preservation and show that there exists a simple process
for verifying whether $\mappingsOriginal$ is safe w.r.t. $\mappingsOriginal_{\viewSchema}$.

\begin{definition}\label{definition:preservation}
    A mapping ${\mappingsOriginal_2 = (\source,\target_2,\Sigma_2)}$ 
    \emph{preserves the privacy of a mapping ${\mappingsOriginal_1 = (\source,\target_1,\Sigma_1)}$ on all instances of $\source$}, if 
    for each constants-free CQ $\secret$ over $\source$, if $\secret$ is not disclosed by $\mappingsOriginal_1$  
    on any instance of $\source$, then $\secret$ is not disclosed by $\mappingsOriginal_2$ on any instance of $\source$.       
\end{definition}

\begin{restatable}{theorem}{thmpreservationkernel}\label{theorem:preservation_under_kernel_property}
    A mapping ${\mappingsOriginal_2 = (\source,\target_2,\Sigma_2)}$ 
    \emph{preserves the privacy of a mapping ${\mappingsOriginal_1 = (\source,\target_1,\Sigma_1)}$ 
    on all instances of $\source$}, if and only if there exists a homomorphism $h$ from 
    $\visins{\Sigma_2}{\source}$ into $\visins{\Sigma_1}{\source}$,
    such that ${h(\ast) = \ast}$.
\end{restatable} 

\begin{proof}
	(Sketch) First we show that the following holds
	
	\begin{lemma}\label{proposition:disclosure}
	    A mapping ${\mappingsOriginal = (\source,\target,\Sigma)}$ 
	    \emph{does not disclose a constants-free CQ $\secret$ over $\source$ on any instance of $\source$}, iff 
	    ${\vec{\ast} \not\in p(J)}$, where ${J = \visins{\Sigma_{st}}{\source}}$. 
	\end{lemma}
	\begin{proof}
		By adapting the proof technique of Theorem 16 from \cite{MCCK17}, we can show that ${J = \visins{\Sigma_{st}}{\source}}$ is a \emph{universal} source instance $\visins{\Sigma}{\source}$ 
		satisfying the following property: for each pair of source instances $I$ and $I'$, such that $I'$ is indistinguishable from $I$ 
		w.r.t. the mapping $\mappingsOriginal$, there exists a homomorphism $h$ from $I'$ into $\visins{\Sigma}{\source}$ 
		mapping each schema constant into the critical constant $\ast$. 
		Due to the existence of a homomorphism $h$ from $I'$ into $\visins{\Sigma}{\source}$, for each pair of indistinguishable source instances $I$ and $I'$, 
		we can see that if ${\vec{\ast} \not\in p(J)}$ for a constants-free CQ $p$, then ${p(I') = \emptyset}$. 
		Due to the above and due to Definition~\ref{definition:disclosure}, it follows that 
		${\mappingsOriginal = (\source,\target,\Sigma)}$ 
	    	does not disclose a constants-free CQ $\secret$ over $\source$ on any instance of $\source$. 
	\end{proof}

	Lemma~\ref{proposition:disclosure} states that, in order to check if a constants-free CQ 
	is safe according to Definition~\ref{definition:disclosure}, we need to check if the critical tuple is 
	among the answers to $\secret$ over the instance computed by ${\vchase{\Sigma}{\source}}$. 
	Next, we show the following lemma. 

    \begin{lemma}\label{theorem:boolean-query-preservation}
        Given two instances $I_1$ and $I_2$, the following are equivalent
        \begin{enumerate}
             \item for each CQ $\secret$, if ${\vec{u} \in \secret(I_1)}$, then ${\vec{u} \in \secret(I_2)}$, where $\vec{u}$ is a vector of constants 
    
             \item there exists a homomorphism from $I_1$ to $I_2$ preserving the constants of $I_1$
        \end{enumerate}    
    \end{lemma}
     
    \begin{proof}[Proof of Lemma~\ref{theorem:boolean-query-preservation}] 
        (2)$\Rightarrow$(1). Suppose that there exists a homomorphism $h$ 
        from $I_1$ to $I_2$ preserving the constants of $I_1$.
        Suppose also that ${\vec{u} \in \secret(I_1)}$, with $\secret$ being a CQ.
        This means that there exists a homomorphism $h_1$ from $\secret$ into $I_1$ mapping 
        each free variable $x_i$ of $\secret$ into $u_i$, for each ${1 \leq i \leq n}$, 
        where $n$ is the number of free variables of $\secret$. 
        Since the composition of two homomorphisms is a homomorphism
        and since $h$ preserves the constants of $I_1$ due to the base assumptions,  
        this means that ${h \circ h_1}$ is a homomorphism from 
        $\secret$ into $I_2$ mapping each free variable $x_i$ of $\secret$ into $t_i$, 
        for each ${1 \leq i \leq n}$.
        This completes this part of the proof.   
    
        (1)$\Rightarrow$(2). 
        Let $\secret_1$ be a CQ formed by creating  
        a non-ground atom ${R(y_1,\dots,y_n)}$ for each ground atom ${R(u_1,\dots,u_n) \in I_1}$,  
        by taking the conjunction of these non-ground atoms and by converting into an existentially
        quantified variable every variable created out of some labelled null. 
        Let $\vec{x}$ denote the free variables of $\secret_1$ and let ${n = |\vec{x}|}$. 
        From the above, it follows that there exists a homomorphism $h_1$ 
        from $\secret_1$ into $I_1$ mapping each ${x_i \in \vec{x}}$ 
        into some constant occurring in $I_1$. Let ${\vec{u} \in \secret_1(I_1)}$. 
        From (1), it follows that ${\vec{u} \in \secret_1(I_2)}$ and,
        hence, there exists a homomorphism $h_2$ from $p_1$ into $I_2$
        mapping each ${x_i \in \vec{x}}$ into $u_i$, for each ${1 \leq i \leq n}$.
        Since $h_1$ ranges over all constants of $I_1$ and since  
        ${h_1(x_i) = h_2(x_i)}$ holds for each ${1 \leq i \leq n}$, 
        it follows that there exists a homomorphism from $I_1$ to $I_2$ preserving the constants of $I_1$.         
        This completes the second part of the proof.    
      \end{proof}
\newpage
    Lemma~\ref{theorem:boolean-query-preservation} can be restated as follows
    \begin{lemma}\label{theorem:boolean-query-preservation-negation}
        Given two instances $I_1$ and $I_2$, the following are equivalent
        \begin{enumerate}
            \item for each CQ $\secret$, if ${\vec{t} \not\in p(I_2)}$, then ${\vec{t} \not\in p(I_1)}$
    
            \item there exists a homomorphism from $I_1$ to $I_2$  
        \end{enumerate}   
    \end{lemma}
    
    We are now ready to return to the main part of the proof.  
    Given a CQ $\secret$ over a source schema $\source$, and a mapping $\mappingsOriginal$
    defined as the triple $(\source,\target,\Sigma)$, where $\target$ is a target schema and 
    $\Sigma$ is a set of s-t dependencies,
    we know from Proposition~\ref{proposition:disclosure} that
    if $\mappingsOriginal$ discloses $\secret$ on some instance of $\source$, 
    then there exists a homomorphism of $\secret$ into $\vchase{\Sigma}{\source}$ mapping the free variables of 
    $\secret$ into the critical constant $\ast$.
    
    From the above, we know that $\mappingsOriginal_2$ does not preserve the privacy of $\mappingsOriginal_1$ 
    if there exists a CQ $\secret$ over $\source$, 
    such that ${\vec{\ast} \not \in J_1}$ and ${\vec{\ast} \in J_2}$, 
    where ${J_1 = \visins{\Sigma_1}{\source}}$ and ${J_2 = \visins{\Sigma_2}{\source}}$.
    We will now prove that $\mappingsOriginal_2$ preserves the privacy of $\mappingsOriginal_1$ iff there exists a
    homomorphism from $J_2$ into $J_1$ that preserves the critical constant $\ast$. 
    This will be referred to as conjecture ($C$).  
    
    ($\Rightarrow$) If $\mappingsOriginal_2$ preserves the privacy of $\mappingsOriginal_1$, then for each CQ 
    $\secret$, if ${\vec{\ast} \not\in \secret(J_1)}$, then ${\vec{\ast} \not\in \secret(J_2)}$.  
    From the above and from Lemma~\ref{theorem:boolean-query-preservation-negation}, it follows that 
    there exists a homomorphism ${\phi:J_2 \rightarrow J_1}$, such that ${\phi(\ast) = \ast}$.  
    
    ($\Leftarrow$) The proof proceeds by contradiction. Assume that there exists a
    homomorphism $h$ from $J_2$ into $J_1$ preserving $\ast$, 
    but $\mappingsOriginal_2$ does not preserve the privacy of $\mappingsOriginal_1$. 
    We will refer to this assumption as assumption ($A_1$).  
    From assumption ($A_1$) and the discussion above it follows 
    that there exists a CQ $\secret$ over $\source$ such that 
    ${\vec{\ast} \not\in \secret(J_1)}$ and ${\vec{\ast} \in \secret(J_2)}$. 
    Let $h_2$ be the homomorphism from $\secret$ into $J_2$ mapping its free variables into $\ast$.
    Since the composition of two homomorphisms is a homomorphism, this means that 
    $h \circ h_2$ is a homomorphism from $\secret$ into $J_1$ mapping its free variables into $\ast$, 
    i.e., ${\vec{\ast} \in \secret(J_1)}$. This contradicts our original assumption and hence concludes the proof of conjecture ($C$).
    Conjecture ($C$) witnesses the decidability of the instance-independent privacy preservation problem: 
    in order to verify whether $\mappingsOriginal_2$ preserves the privacy of $\mappingsOriginal_1$ we only need to check if there exists 
    a homomorphism ${\phi:\visins{\Sigma_2}{\source} \rightarrow \visins{\Sigma_1}{\source}}$, such that ${\phi(\ast) = \ast}$.  
\end{proof}

Theorem~\ref{theorem:preservation_under_kernel_property} states that in order to 
verify whether $\mappingsOriginal_2$ is safe w.r.t. $\mappingsOriginal_1$, 
we need to compute $\visins{\Sigma_1}{\source}$ and $\visins{\Sigma_2}{\source}$ 
and check if there exists a homomorphism from the second instance into the first one 
that maps $\ast$ into itself. 
If there exists such a homomorphism, we say that $\visins{\Sigma_1}{\source}$ 
is \emph{safe} w.r.t. $\visins{\Sigma_2}{\source}$, or simply safe, and we say that it is unsafe otherwise.  
 
\begin{example}
    Continuing with Example~\ref{example:running:intro}, we can see that the s-t tgds are not safe w.r.t. 
    the policy views according to Theorem~\ref{theorem:preservation_under_kernel_property}, since there does not exist a homomorphism from
    the instance $\visins{\Sigma_{st}}{\source}$ into
    the instance $\visins{\policyViews}{\source}$.
    This means that there exists information which is disclosed by 
    $\Sigma_{st}$ in some instance that satisfies $\Sigma_{st}$, but it is not disclosed by $\policyViews$.  
    Indeed, from ${\school(\lnu''_{\sid},\lnu''_{\name},\ast,\lnu'_{\post})}$ and ${\onc(\lnu''_{\sid},\lnu''_{\tr},\lnu''_{\pr})}$,
    we can see that we can potentially leak the identity of a student who has been to an oncology department. 
    This can happen if there exists only one student in the school coming from a specific ethnicity group
    and this ethnicity group is returned by $\mu_s$. 
    Please note that the policy views are safe w.r.t. this leak. Indeed, it is impossible to derive 
    this information through reasoning over the returned tuples under the input instance and the views 
    $\viewPred_3$ and $\viewPred_4$. 
    
    Furthermore, by looking at the facts ${\pat(\lnu_{\pid},\lnu_{\name},\ast,\ast)}$ and ${\hospn(\lnu_{\pid},\ast)}$, 
    we can see that we can potentially leak the identity and the disease of a patient who has been admitted to some hospital 
    in the north of UK. This can happen if 
    there exists only one patient who relates to the county and the ethnicity group returned by $\mu_e$ and $\mu_c$.  
    Note that the policy views $\viewPred_1$ and $\viewPred_2$ do not leak this information, 
    since it is impossible to obtain the county and the ethnicity group of an NHS patient at the same time.
\end{example}

%% file: repair.tex
\section{Repairing unsafe mappings} \label{sec:repair}


In Section~\ref{sec:privacy} we presented our privacy preservation protocol
and a technique for verifying whether a mapping is safe w.r.t. another one, 
over all source instances. This section presents an algorithm for repairing an unsafe mapping 
$\mappingsOriginal$ w.r.t. a set of policy views $\policyViews$.

Algorithm~\ref{alg:repair} summarizes the steps of the proposed algorithm. The inputs to 
it are, apart from $\Sigma$ and $\policyViews$, a positive integer $n$ 
which will be used during the second step of the repairing process and a preference 
mechanism $\pref$ for ranking the possible repairs. In the simplest scenario, 
the preference mechanism implements a fixed function for ranking the different repairs. 
However, it can also employ supervised learning techniques 
in order to progressively learn the user preferences by looking at his prior decisions.

Since a mapping $\mappingsOriginal$ is safe w.r.t. $\policyViews$
if the instance $\visins{\Sigma}{\source}$ is safe according to Theorem~\ref{theorem:preservation_under_kernel_property}, 
Algorithm~\ref{alg:repair} rewrites the tgds in $\mappingsOriginal$, 
such that the derived visible chase instances are safe.
The rewriting takes place in two steps. 
The first step rewrites $\Sigma$ into a \emph{partially-safe} set of s-t dependencies $\Sigma_1$, 
while the second step rewrites the output of the first one into a new set of s-t 
dependencies $\Sigma_2$, such that $\visins{\Sigma_2}{\source}$ is safe.
As we will explain later on, partial-safety ensures that the intermediate instance $I_1$ produced by $\vchase{\Sigma_1}{\source}$ 
at line~\ref{bagchase:back} of Algorithm~\ref{alg:bagchase} is safe, but it does not provide strong privacy guarantees. 
The benefit of this two-step approach is that it allows repairing one or a small set of dependencies at a time.  

\begin{algorithm}[t]
\scriptsize
\caption{$\repair(\Sigma,\policyViews,\pref,n)$}\label{alg:repair}
\begin{algorithmic}[1]
    \State $\Sigma_1 \defeq \fphase(\Sigma,\policyViews,\pref)$      \label{alg:repair:first}                   
    \State $\Sigma_2 \defeq \sphase(\Sigma_1,\policyViews,\pref,n)$  \label{alg:repair:second}        
    \State \textbf{return} $\Sigma_2$ 
\end{algorithmic}
\end{algorithm}

\subsection{Computing partially-safe mappings}\label{sec:repair:partial-safety}

Since the problem of safety is reduced to the problem of 
checking for a homomorphism from $\visins{\Sigma}{\source}$ into 
$\visins{\policyViews}{\source}$, a first test 
towards checking for such a homomorphism is to look if the mappings in 
$\Sigma$ would lead to such a homomorphism or not. 
For instance, by looking at $\mu_s$ in Example~\ref{example:running:intro} it is easy to see that 
it leaks sensitive information, since it involves a join between students and
oncology departments, which does not occur in 
$\visins{\policyViews}{\source}$.

\begin{definition}\label{definition:potential-safety}
    A mapping ${\mappingsOriginal = (\source,\target,\Sigma)}$ is 
    \emph{partially-safe w.r.t.\\ ${\mappingsOriginal_{\viewSchema} = (\source,\viewSchema,\policyViews)}$ 
    on all instances of $\ \source$}, if there exists a homomorphism from
    ${\chase{\inverse{\Sigma}}{\critical{\target}} \setminus \critical{\target}}$ 
    into $\visins{\policyViews}{\source}$. 
\end{definition}

From Algorithm~\ref{alg:bagchase}, it follows that 
$\Sigma$ is partially-safe iff the intermediate instance $I_1$ 
computed by $\vchase{\Sigma}{\source}$ is safe.  

\begin{proposition}\label{proposition:potential-safety}
    A mapping ${\mappingsOriginal = (\source,\target,\Sigma)}$ 
    is \emph{partially-safe w.r.t.
    ${\mappingsOriginal_{\viewSchema} = (\source,\viewSchema,\policyViews)}$ 
    on all instances of $\source$}, if 
    for each ${\mu \in \Sigma}$, there exists a homomorphism 
    from $\body(\mu)$ into $\visins{\policyViews}{\source}$ mapping each 
    ${x \in \frontier{\mu}}$ into the critical constant $\ast$. 
\end{proposition}  

Note that according to Proposition~\ref{proposition:potential-safety}, in
our running example  
$\Sigma_{st}$ would be 
partially-safe, if $\mu_s \not\in \Sigma_{st}$, then since there exist homomorphisms from the bodies of 
$\mu_s$ and $\mu_c$ into $\visins{\policyViews}{\source}$, 
mapping their exported variables into $\ast$. 
It is also easy to show the following 

\begin{remark}\label{remark:potential-safety}
    A mapping ${\mappingsOriginal = (\source,\target,\Sigma)}$ 
    is \emph{safe w.r.t.
    ${\mappingsOriginal_{\viewSchema} = (\source,\viewSchema,\policyViews)}$ 
    on all instances of $\source$}, only if it is partially-safe w.r.t.
    $\mappingsOriginal_{\viewSchema}$ on all instances of $\source$. 
\end{remark}  

Proposition~\ref{proposition:potential-safety} presents a quite convenient, 
yet somewhat expected, finding: in order to obtain a partially-safe mapping, 
it suffices to repair each s-t dependency \emph{independently of the others}. 
Furthermore, the repair of each ${\mu \in \Sigma}$ involves breaking joins and hiding exported variables, 
such that the repaired dependency $\mu_r$ satisfies the criterion in Proposition~\ref{proposition:potential-safety}.

We make use of the result of Proposition~\ref{proposition:potential-safety} 
in Algorithm~\ref{alg:frepair}. Algorithm~\ref{alg:frepair} obtains, for each 
${\mu \in \Sigma}$, a set of rewritings $\mathcal{R}_{\mu}$, 
out of which we will choose the best rewriting according to $\pref$. 
The set $\mathcal{R}_{\mu}$ consists of \emph{all} rewritings
that differ from $\mu$ w.r.t. the variable repetitions 
in the bodies of the rules and the exported variables. 
For performance reasons, we do not examine rewritings 
that introduce atoms in the bodies of the rules. 
However, this does not compromise the \emph{completeness} of  
Algorithm~\ref{alg:repair} as we show at end of this section.
Below, we present the steps of Algorithm~\ref{alg:frepair}.

For each s-t tgd $\mu$ and for each atom ${B \in \body(\mu)}$, 
Algorithm~\ref{alg:frepair} constructs a fresh atom $C$ 
and adds $C$ to a set $\mathcal{C}$. The set of atoms $\mathcal{C}$
provides us with the means to identify all repairs of $\mu$ 
that involve breaking joins and hiding exported variables.
In particular, each homomorphism $\xi$ from $\mathcal{C}$ into $\visins{\policyViews}{\source}$
corresponds to one repair of $\mu$. In lines~\ref{alg:B:begin}--\ref{alg:B:end}, 
Algorithm~\ref{alg:frepair} modifies each atom ${B \in \body(\mu)}$ by taking into account 
prior body atom modifications. The prior modifications are accumulated in 
the relation $\rho$ and the mapping $\psi$. The relation $\rho$ keeps for each variable $x$ from $\body(\mu)$, 
the fresh variables that were used to replace $x$
during prior steps of the repairing process, while 
$\psi$ is a substitution from the partially 
repaired body into $\visins{\policyViews}{\source}$. In particular,
at the end of the $i$-th iteration of the loop in line~\ref{alg:B:begin}, 
$\psi$ holds the substitution from the first repaired $i$ atoms from 
$\body(\mu)$ into $\visins{\policyViews}{\source}$.   
We adopt this approach instead of replacing variable 
$x$ in position $p$ always by a fresh variable, 
in order to minimize the number of the joins we break.

Below, we describe how Algorithm~\ref{alg:frepair} modifies each body atom of $\mu$, 
w.r.t. a homomorphism $\xi$, lines~\ref{alg:xi:begin}--\ref{alg:xi:end}.   
Let ${C = \nu(B)}$ be the fresh body atom that was constructed out of $B$ in line~\ref{alg:C:init}.   
For each atom ${B \in \body(\mu)}$ and for each ${p \in \positions{B}}$, 
if the variable $y$ in position $p$ of $C$ is not mapped to the critical constant $\ast$
via $\xi$ and $\pos{B}{p}$ is a exported variable, this means that 
\emph{the variable sitting in position $p$ of $B$ should not be exported}
(see first condition in line~\ref{alg:x:begin}). 
Similarly, if the variable sitting in position $p$ of $B$ is mapped to 
a different constant than the one that $y$ maps via $\xi$, then  
this means that \emph{the variable sitting in position $p$ of $B$ 
introduces an unsafe join}
(see second condition in line~\ref{alg:x:begin}).
In the presence of these violations, we must replace variable $x$ in position $p$ of $B$,
either by a variable that was used in a prior step of the repairing process,
line \ref{alg:x:prior:begin}--\ref{alg:x:prior:end}), or by a fresh variable,
lines~\ref{alg:x:fresh:begin}--\ref{alg:x:fresh:end}.
Otherwise, if there is no violation so far, then we add the mapping
${\{x \mapsto \xi(y)\}}$ to $\psi$, if it is not already there,
lines~\ref{alg:no-violation:start}--\ref{alg:no-violation:end}.
Finally, the algorithm chooses the best repair according to the preference function, 
lines~\ref{alg:pref:start}--\ref{alg:pref:end}.

\begin{algorithm}[tb!]
\scriptsize
\caption{$\fphase(\Sigma,\policyViews,\pref)$}\label{alg:frepair}
\begin{algorithmic}[1]
    \For{\textbf{each} $\mu \in \Sigma$}
          \State $\nu \defeq \emptyset$, $\mathcal{C} \defeq \emptyset$
          \For{\textbf{each} $B \in \body(\mu)$, where $B = R(\vec{x})$}                                               						 			\label{alg:C:start}
            \State \textbf{create} a vector of fresh variables $\vec{y}$
            \State \textbf{create} the atom $C = R(\vec{y})$                                                           											 \label{alg:C:init}
            \State \textbf{add} $(B,C)$ to $\nu$
            \State \textbf{add} $C$ to $\mathcal{C}$
         \EndFor                                                                                                        														\label{alg:C:end}
         
         \State $\mathcal{R}_{\mu} := \emptyset$

         \For{\textbf{each} homomorphism $\xi : \mathcal{C} \to \visins{\policyViews}{\source}$}                       							 \label{alg:xi:begin}
             \State $\rho := \emptyset$, $\psi := \emptyset$
             \State $\mu_r := \mu$
             \For{\textbf{each} $B \in \body(\mu_r)$}                                                                   											\label{alg:B:begin}                                                         
                 \State $C=\nu(B)$																										\label{alg:C:extract}                                                         
                 \For{\textbf{each} $p \in \positions{B}$}                                                              											\label{alg:p:begin}   
                     \State $x = \pos{B}{p}$, $y = \pos{C}{p}$                                                          
                     \If{$x \in \frontier{\mu}$ \textbf{and} $* \neq \xi(y)$ \textbf{or} $x \in \dom{\psi}$  \textbf{and} $\psi(x) \neq \xi(y)$}         \label{alg:x:begin}   
                         \If{$\exists x'$ s.t. $(x,x') \in \rho$ and $\psi(x') = \xi(y)$}                               											\label{alg:x:prior:begin}
                             \State $\pos{B}{p} = x'$                                                                   													\label{alg:x:prior:end}
                         \Else                                                                                          														\label{alg:x:fresh:begin}
                             \State \textbf{create} a fresh variable $x'$                                               
                             \State \textbf{add} $(x,x')$ to $\rho$
                             \State \textbf{add} $\{x' \mapsto \xi(y)\}$ to $\psi$
                             \State $\pos{B}{p} = x'$
                         \EndIf                                                                                         														\label{alg:x:fresh:end}
                     \ElsIf{$x \not \in \dom{\psi}$}                                                                    												\label{alg:no-violation:start}
                         \State \textbf{add} $\{x \mapsto \xi(y)\}$ to $\psi$                                          									 		\label{alg:no-violation:end}
                     \EndIf																													\label{alg:x:end}
                 \EndFor                                                                                                														\label{alg:p:end}   
             \EndFor                                                                                                    														\label{alg:B:end}
             \If{$\mu_r \neq \mu$}
                 \State \textbf{add} $\mu_r$ to $\mathcal{R}_{\mu}$
             \EndIf
         \EndFor                                                                                                        														\label{alg:xi:end}
         
         \If{$\mathcal{R}_{\mu} \neq \emptyset$}                                                                    f    											\label{alg:pref:start}
             \State \textbf{choose} the best repair $\mu_r$ of $\mu$ from $\mathcal{R}_{\mu}$ based on $\pref$
             \State \textbf{remove} $\mu$ from $\Sigma$
             \State \textbf{add} $\mu_r$ to $\Sigma$
         \EndIf                                                                                                         														\label{alg:pref:end}
    \EndFor
    \State \textbf{return} $\Sigma$
\end{algorithmic}
\end{algorithm}

\begin{example}\label{example:frepair}
    We demonstrate an example of Algorithm~\ref{alg:frepair}. 
    
    Since Algorithm~\ref{alg:frepair} focuses on $\visins{\policyViews}{\source}$
    overlooking the actual views in $\policyViews$, we will not explicitly define 
    $\policyViews$. Instead, we will only assume that the visible chase computes the instance  
    \begin{flalign}
        \visins{\policyViews}{\source} = \{\R_1(\ast,\lnu_1,\lnu_2), \Sm_1(\lnu_1,\lnu_2, \lnu_2), \Sm_1(\lnu_1,\lnu_3,\ast), \Sm_1(\lnu_1,\ast,\ast)\} \nonumber 
    \end{flalign}
    where $\lnu_1$--$\lnu_3$ are labeled nulls. 
    Consider also the mapping $\mappingsOriginal$ consisting of the following s-t dependency 
    \begin{align}
        \R_1(x, y, z) \wedge \Sm_1(y, z, z) \rightarrow \T_1(x,z) \tag{$\mu_1$} \nonumber 
    \end{align}                             
    Note that $\mappingsOriginal$ is not partially-safe. 
    Algorithm~\ref{alg:frepair} computes two repairs for $\mu_1$ by applying the steps described below. 
    First, it computes the atoms ${\R_1(x_1, x_2, x_3)}$ ${\Sm_1(x_4, x_5, x_6)}$ and adds them to 
    $\mathcal{C}$, lines~\ref{alg:C:start}--\ref{alg:C:end}.
    Then, it identifies the following three homomorphisms from $\mathcal{C}$ into $\visins{\policyViews}{\source}$:
    \begin{align}
        \xi_1 &= \{x_1 \mapsto \ast, x_2 \mapsto \lnu_1, x_3 \mapsto \lnu_2, x_4 \mapsto \lnu_1, x_5 \mapsto \lnu_2, x_6 \mapsto \lnu_2\} \nonumber \\
        \xi_2 &= \{x_1 \mapsto \ast, x_2 \mapsto \lnu_1, x_3 \mapsto \lnu_2, x_4 \mapsto \lnu_1, x_5 \mapsto \lnu_3, x_6 \mapsto \ast\}   \nonumber \\
        \xi_3 &= \{x_1 \mapsto \ast, x_2 \mapsto \lnu_1, x_3 \mapsto \lnu_2, x_4 \mapsto \lnu_1, x_5 \mapsto \ast, x_6 \mapsto \ast\}     \nonumber 
    \end{align}                 
    From $\xi_1$, we can see that the joins in the body of $\mu_1$ are safe; however, it is unsafe to export $z$. 
    From $\xi_2$, we can see that is safe to reveal the third position of $\Sm_1$; however, it is unsafe to join the 
    second and the third position of $\Sm_1$.     
    Algorithm~\ref{alg:frepair} then iterates over $\xi_1$ and $\xi_2$, line~\ref{alg:xi:begin}. 
    When ${B = \R_1(x, y, z)}$ and ${p < 3}$, Algorithm~\ref{alg:frepair} computes 
    $\psi$ to $\{x \mapsto \ast, y \mapsto \lnu_1\}$, 
    since there is no violation according to line~\ref{alg:x:begin}.  
    When ${B = \R_1(x, y, z)}$ and $p=3$, however, a violation is detected.  
    This is due to the facts that $z$ is an exported variable and ${\xi(x_3) = \lnu_2}$.
    Algorithm~\ref{alg:frepair} tackles this violation by creating a fresh variable $z_1$.
    Then, it adds the relation $(z,z_1)$ to $\rho$,  
    replaces $z$ in $\pos{B}{3}$ by $z_1$ and adds the mapping  
    ${\{z_1 \mapsto \lnu_2\}}$ to $\psi$, lines~\ref{alg:x:fresh:begin}--\ref{alg:x:fresh:end}.   
    Algorithm~\ref{alg:frepair} then considers $\Sm_1(y, z, z)$.
    When $p = 1$, no violation is encountered, since ${\psi(y) = \xi_1(x_4)}$.
    However, when $p = 2$, a homomorphism violation is encountered, since
    $z$ is an exported variable and since ${\xi(x_3) = \lnu_2}$.
    Since ${(z,z_1) \in \rho}$ and   
    ${\psi(z_1) = \xi_1(x_5)}$, Algorithm~\ref{alg:frepair} replaces
    $z$ in the second position of $\Sm_1(y, z, z)$ by $z_1$, 
    line~\ref{alg:x:fresh:begin}. By applying a similar reasoning, 
    we can see that the variable $z$ siting in $\pos{\Sm_1(y, z, z)}{3}$ is also replaced by $z_1$. 
    Hence, the first repair of $\mu$ is 
    \begin{align}
        \R_1(x, y, z_1) \wedge \Sm_1(y, z_1, z_1) \rightarrow \T_1(x) \tag{$r_1$} \nonumber
    \end{align}   
    
    Algorithm~\ref{alg:frepair}, then proceeds by repairing $\mu_1$ based on $\xi_2$. 
    When ${B = \R_1(x, y, z)}$, Algorithm~\ref{alg:frepair} proceeds as described above and
    computes $\psi$ to $\{x \mapsto \ast, y \mapsto \lnu_1, z_1 \mapsto \lnu_2\}$. 
    When ${B = \Sm_1(y, z, z)}$ and $p=1$, then no violation is encountered since
    ${\psi(y) = \xi_1(x_4)}$, while when ${B = \Sm_1(y, z, z)}$ and $p=2$, there is a violation. 
    Since the condition in line~\ref{alg:x:prior:end} is not met, Algorithm~\ref{alg:frepair}
    creates a fresh variable $z_2$ and adds the mapping $\{z_2 \mapsto \lnu_3 \}$ to $\psi$.
    When ${B = \Sm_1(y, z, z)}$ and $p=3$, then no violation is met, since ${z \in \frontier{\mu}}$ and ${\xi_2(x_6) = \ast}$. 
    Hence, the second repair of $\mu_1$ is 
    \begin{align}
        \R_1(x, y, z_1) \wedge \Sm_1(y, z_2, z) \rightarrow \T_1(x,z) \tag{$r_2$} \nonumber
    \end{align}   
    
    Finally, we can see that the repair for $\mu_1$ w.r.t. $\xi_3$ is 
    \begin{align}
        \R_1(x, y, z_1) \wedge \Sm_1(y, z, z) \rightarrow \T_1(x,z) \tag{$r_3$} \nonumber
    \end{align}
\end{example}

\begin{restatable}{proposition}{thmfrepair}\label{proposition:frepair}
    For any ${\mappingsOriginal = (\source,\target,\Sigma)}$, 
    any ${\mappingsOriginal_\viewSchema = (\source,\viewSchema,\policyViews)}$
    and any preference function $\pref$, Algorithm $\fphase$ returns a mapping 
    ${\mappingsOriginal' = (\source,\target,\Sigma')}$ that is partially-safe w.r.t. 
    $\mappingsOriginal_\viewSchema$ on all instances of $\source$.
\end{restatable}    
\begin{proof}
	(Sketch) From Proposition~\ref{proposition:potential-safety}, 
    	a mapping ${\mappingsOriginal = (\source,\target,\Sigma)}$ 
    	is partially-safe w.r.t.
    	${\mappingsOriginal_{\viewSchema} = (\source,\viewSchema,\policyViews)}$ 
    	on all instances of $\source$, if 
    	for each ${\mu \in \Sigma}$, there exists a homomorphism 
    	from $\body(\mu)$ into $\visins{\policyViews}{\source}$ mapping each 
    	${x \in \frontier{\mu}}$ into the critical constant $\ast$. 
    	Since for each ${\mu \in \Sigma}$ $\fphase$ computes a set of repaired tgds ${\mathcal{R}_{\mu}}$, 
    	it follows that Proposition~\ref{proposition:frepair} holds, if such a homomorphism 
    	exists, for each repaired tgd in ${\mathcal{R}_{\mu}}$. The proof proceeds as follows. Let 
    	$\mu_r^i$ and $\psi^i$ denote the repaired s-t tgd and the homomorphism $\psi$ computed at the end of each iteration $i$ 
	of the steps in lines~\ref{alg:B:begin}--\ref{alg:B:end} of Algorithm~\ref{alg:frepair}.
	Let also $B^i$ denote the $i$-th atom in ${\body(\mu_r)}$. 
	Since each ${C \in \mathcal{C}}$ is an atom of distinct fresh variables, 
	since $\xi$ is a homomorphism from $\mathcal{C}$ to $\visins{\policyViews}{\source}$
	and since ${\psi(B^i) = \pos{\mu_r}{i}}$, it follows that 
	in order to prove Proposition~\ref{proposition:potential-safety},
	we have to show that the following claim holds, for each ${i \geq 0}$:
	\begin{compactitem}
		\item $\phi$. $\psi^i$ is a homomorphism from the first $i$ atoms in the body of $\mu_r$ into $\visins{\policyViews}{\source}$
		mapping each exported variable occurring in ${B^0,\dots,B^i}$ into the critical constant $\ast$.  
	\end{compactitem}
	For ${i=0}$, $\phi$ trivially holds. For ${i+1}$ and assuming that $\phi$ 
	holds for $i$ let ${C^{i+1} = \nu(B^{i+1})}$, line~\ref{alg:C:extract}.  
	The proof of claim $\phi$ depends upon the proof of the following claim, for each iteration ${p \geq 0}$ of the steps in 
	lines~\ref{alg:p:begin}--\ref{alg:p:end}:
	\begin{compactitem}
		\item $\theta$. $\psi^{i+1}(\pos{B^{i+1}}{p}) = \xi(y)$, where ${y = \pos{C^{i+1}}{p}}$.  
	\end{compactitem}
	The claim $\theta$ trivially holds for ${p=0}$, while for ${p>0}$, it directly follows from the steps in lines~\ref{alg:x:begin}--\ref{alg:x:end}. 
	Since $\phi$ holds for $i$, since the steps in lines~\ref{alg:x:begin}--\ref{alg:x:end} do not modify the variable mappings in $\psi^i$
	and due to $\theta$, 	it follows that $\phi$ holds for ${i+1}$, concluding the proof of Proposition~\ref{proposition:potential-safety}.
\end{proof}

\subsection{Computing safe mappings}\label{sec:repair:safety}

Unifications of one or more labeled nulls occurring in $I_1$ 
with the critical constant $\ast$, might lead to unsafe instances.
Consider, for instance, a simplified variant of Example~\ref{example:running:intro}, where $\Sigma_{st}$ comprises only 
$\mu_e$ and $\mu_c$. Both $\mu_e$ and $\mu_c$ are partially-safe, as we have explained above. 
However, the unification of the labeled nulls $\lnu_{\name}$ and $\lnu_{\post}$ produces an unsafe instance. 
Algorithm~\ref{alg:sphase} aims at repairing the output of the previous step, such that no unsafe unification 
of a labeled null with $\ast$ takes place. 

Consider again the simplified variant of $\Sigma_{st}$ from above.  
Since $\Sigma_{st}$ is partially-safe, it suffices to look for homomorphism violations in $I_i$, for ${i \geq 1}$.  
A first observation is that the homomorphism violations are ``sitting" within the bags. 
This is due to the fact that each bag stores \emph{all} 
the facts associated with the bodies of one or more s-t tgds from $\Sigma_{st}$.
A second observation is that one way for preventing unsafe unifications is to hide exported variables. 
For example, let us focus on the unsafe unification of $\lnu_{\ethn}$ with $\ast$. 
This unification takes place due to $\epsilon_1$, 
which in turn has been created due to the fact that $\ethn$ is an exported variable in $\mu_e$. 
By hiding the exported variable $\ethn$ from $\mu_e$, 
we actually prevent the creation of $\epsilon_1$ and hence, 
we block the unsafe unification of $\ethn$ with $\ast$.     
Hiding exported variables is one way for preventing unsafe unifications with the critical constant. 
Another way for preventing unsafe unifications is to break joins in the bodies of the rules. 

\begin{example}\label{example:srepair:modifyBody}
    This example demonstrates a second approach for preventing unsafe labeled null unifications.
    
    Consider a set of policy views $\policyViews$ leading to the following instance
    ${\visins{\policyViews}{\source} = \{\R_1(\lnu_1,\lnu_1,\ast), \R_1(\ast,\ast,\lnu_2), \Sm_1(\ast)\}}$,
    where $\lnu_1$ and $\lnu_2$ are labelled nulls. 
    Consider also the mapping $\mappingsOriginal$ consisting of the following s-t dependencies:
    \begin{align}
        \R_1(x, x, y) \wedge \Sm_1(y) &\rightarrow \T_1(y) \tag{$\mu_2$} \\ 
        \R_1(x, x, y)                 &\rightarrow \T_2(x) \tag{$\mu_3$}    
    \end{align}  
    It is easy to see that $\mappingsOriginal$ is partially-safe, but unsafe in overall. Indeed, ${\visins{\Sigma}{\source}}$
    will consist of the following bags (for presentation purposes, we adopt the notation from Example~\ref{example:running:visChase}):
    \begin{align}
        \T_1(\ast)      &\xrightarrow{\langle \inverse{\mu_2}, \theta_1 \rangle}  \R_1(\lnu_3,\lnu_3, \ast), \Sm_1(\ast)                  \nonumber \\
        \T_2(\ast)      &\xrightarrow{\langle \inverse{\mu_3}, \theta_2 \rangle}  \R_1(\ast,\ast,\lnu_4)                                  \nonumber \\
        \R_1(\lnu_3,\lnu_3, \ast), \Sm_1(\ast)  &\xrightarrow{\langle \epsilon_3, \theta_3 \rangle} \R_1(\ast, \ast, \ast), \Sm_1(\ast)   \nonumber
    \end{align}
    where ${\epsilon_3 \defeq \R_1(x, x, y) \rightarrow x = \ast}$, 
    ${\theta_1 = \{y \mapsto \ast\}}$, ${\theta_2 = \{x \mapsto \ast\}}$ and 
    ${\theta_3 = \{x \mapsto \lnu_3, y \mapsto \ast\}}$.  
    Note that $\epsilon_3$ has been created out of $\mu_3$, 
    since there exists a homomorphism from $\body(\mu_3)$ into 
    ${\R_1(\lnu_3,\lnu_3, \ast)}$ mapping the exported variable $x$ into $\lnu_3$.  
    
    One approach for preventing the unsafe unification of $\lnu_3$ with $\ast$ is 
    to hide the exported variable $x$ from $\mu_3$. By doing this, we block 
    the creation of $\varepsilon$, and hence the unsafe unification. 
    
    A second approach is to keep $x$ as an exported variable in $\mu_3$, but  
    modify the body of $\mu_2$
    by breaking the join between the first and the second position of $\R_1$
    \begin{align}
        \R_1(x, z, y) \wedge \Sm_1(y) &\rightarrow \T_1(y) \tag{$\mu'_2$} \nonumber 
    \end{align}
    By doing this, we prevent the creation of $\varepsilon$, since the instance computed at line~\ref{bagchase:back}  
    of Algorithm~\ref{alg:bagchase} would consist of the facts $\R_1(\lnu_3,\lnu_5, \ast),$ $\R_1(\ast,\ast,\lnu_4),$ $\Sm_1(\ast)$ 
    and, hence, there would be no homomorphism from $\body(\mu_3)$
    into it. Note that the modification of $\mu_2$ to $\mu'_2$ is safe.
    Intuitively, this holds, since we break joins, and thus, we export less information.
\end{example}

Before presenting Algorithm~\ref{alg:sphase}, we will introduce some new notation. 
The \emph{depth} of each 
bag $\bag$, denoted as $\depth{\bag}$, coincides with the highest derivation depth 
of the facts in $\bag$.
The \emph{support} of a bag $\bag$, denoted as $\supp{\bag}$, is inductively defined as follows:
if ${\depth{\bag} = 1}$, then ${\supp{\bag} = \bag}$; otherwise, 
if ${\depth{\bag} > 1}$, then ${\cup_{\bag' \prec \bag} \supp{\bag'}}$.
Consider an active trigger $h$ for $\delta$ in $I$ leading to the creation of a bag $\bag$. We use the following notation:
${\dependency{\bag} = \delta}$,  ${\sub{\beta} = h}$ and ${\prem{\bag} = h(\body(\delta))}$. 
Two bags $\bag_1$ and $\bag_2$ are candidates for $\elimVars$ if 
${\bag_1 \prec \bag_2}$, ${\depth{\bag_1} = 1}$, ${\depth{\bag_2} = 2}$ 
and there exists at least one repeated variable in the body of $\origin{\bag_1}$.
                 
Algorithm~\ref{alg:sphase} presents an iterative process for repairing a partially-safe $\Sigma$, by 
employing the three ideas we described above: checking for homomorphism violations within each bag 
and preventing unsafe unifications either by hiding exported variable, or by modifying the bodies of the s-t tgds.  
In brief, at each iteration ${i \geq 0}$, the algorithm repairs one or more dependencies from $\Sigma_i$, 
where $\Sigma_0 = \Sigma$, and incrementally computes the visible chase of the new set of dependencies, 
lines~\ref{alg:sphase:i:start}--\ref{alg:sphase:i:end}. Algorithm~\ref{alg:sphase} terminates 
either when the dependencies are safe, or when the maximum number of iterations $n$ is reached, 
line~\ref{alg:sphase:i:end}, in which case it repairs all unsafe dependencies by hiding their 
exported variables. The algorithm starts by initializing $\Sigma_0$ to $\Sigma$, lines~\ref{alg:sphase:Sigma:init}.
Then, at each iteration $i$, it first identifies the lowest depth unsafe bag, 
line~\ref{alg:sphase:bag:init}, and attempts to repair the dependencies from $\Sigma_i$ that 
lead to its creation, lines~\ref{alg:sphase:bag:start}--\ref{alg:sphase:bag:end}. 
If ${i < n}$, it proposes two different repairs for
$\Sigma_i$, one based on hiding exported variables through $\hideF$ (Algorithm~\ref{alg:hideF}), and the second based on 
eliminating joins through $\elimVars$ (Algorithm~\ref{alg:elimVars}), lines~\ref{alg:sphase:smalli:start}--\ref{alg:sphase:smalli:end}.
Algorithm~\ref{alg:sphase} applies the $\elimVars$ if there exist two bags 
in the support of $\bag$ that are candidates for $\elimVars$.
Informally, Algorithm~\ref{alg:sphase} tries to apply $\elimVars$ as early as possible
(condition ${\depth{\bag_1} = 1}$, ${\depth{\bag_2} = 2}$) and when there are one
or more repeated variables in the body of $\origin{\bag_1}$ 
(recall Example~\ref{example:srepair:modifyBody}).
Otherwise, if ${i = n}$, it either applies the function $\hideF$, or it eliminates the s-t tgds 
that are responsible for unsafe unifications. 
     
\begin{algorithm}[tb!]
\scriptsize
\caption{$\sphase(\Sigma,\policyViews,\pref,n)$}\label{alg:sphase}
\begin{algorithmic}[1]
    \State $\Sigma_0 \defeq \Sigma$                                                                                             \label{alg:sphase:Sigma:init} 
    \State $B_0 \defeq \vchase{\Sigma}{\source}$                                                                         \label{alg:sphase:prov:init} 
    \State $i \defeq 0$
    \Do                                                                                                                         \label{alg:sphase:i:start}  
        \State $\Sigma_{i+1} \defeq \Sigma_i$
        \State $cont \defeq \mathbf{false}$ 
        \If{$\exists$ unsafe ${\bag \in B_i}$, s.t. $\depth{\bag} \leq \depth{\bag'}$, $\forall$ unsafe bag ${\bag' \in B_i}$}  \label{alg:sphase:bag:init}                                                                                           \label{alg:sphase:bag:start}
           \State $cont \defeq \mathbf{true}$           
            \If{$i < n$}                                                                                                        \label{alg:sphase:smalli:start}
                 \State ${r_1 \defeq \emptyset}$; $r_2 \defeq \hideF(\bag,\policyViews,\pref)$                                                                                                     \label{alg:sphase:smalli:start}      
                 \If{${\exists \bag_1, \bag_2 \in \supp{\bag}}$, s.t. $\bag_1, \bag_2$ are candidates for $\elimVars$}
                     \State $r_1 \defeq \elimVars(\origin{\beta_1},\origin{\bag_2},\pref)$
                 \EndIf                 
                 \If{$r_1 \neq \emptyset$ and it is preferred over $r_2$ w.r.t. $\pref$} 
                     \State \textbf{remove} $\origin{\bag_1}$ from $\Sigma_{i+1}$ 
                     \State \textbf{add} $r_1$ to $\Sigma_{i+1}$
                 \Else
                     \State \textbf{remove} $\origin{\bag}$ from $\Sigma_{i+1}$ 
                     \State \textbf{add} $r_2$ to $\Sigma_{i+1}$
                 \EndIf
             \Else                                                                                                              \label{alg:sphase:smalli:end}       
                 \If{$\not \exists \bag'$, s.t., $\bag \prec \bag' \in B_i$}                                                    \label{alg:sphase:bigi:start}  
                     \State \textbf{add} $\hideF(\bag,\policyViews,\pref)$ to $\Sigma_{i+1}$
                 \Else \textbf{ remove} $\origin{\bag}$ from $\Sigma_{i+1}$
                 \EndIf                                                                                                         \label{alg:sphase:bigi:end}
             \EndIf
        \EndIf                                                                                                                  \label{alg:sphase:bag:end}
        \State \textbf{compute} $J_{i+1}$ from $\Sigma_i$, $\Sigma_{i+1}$ and $B_{i}$                                           \label{alg:sphase:incremental}
        \State $i = i + 1$
    \doWhile{$cont$ \textbf{and} $i \leq n$}                                                                                    \label{alg:sphase:i:end}
    \State \textbf{return} $\Sigma_{n}$
\end{algorithmic}
\end{algorithm}

\begin{algorithm}[tb!]
\scriptsize
\caption{$\hideF(\bag,\policyViews,\pref)$}\label{alg:hideF}
\begin{algorithmic}[1]
    \State $J := \prem{\bag}$																			\label{alg:hideF:J:begin}
    \State $\nu := \emptyset$
    \For{\textbf{each} $n \in \nulls$ occurring into $J$}                                                            			\label{alg:hideF:nu:begin}
      \State \textbf{add} $\{n \mapsto x\}$ to $\nu$, where $x$ is a fresh variable
    \EndFor                                                                                                          						\label{alg:hideF:nu:end}
  
    \State $\mathcal{R} := \emptyset$
    \For{\textbf{each} $\xi : \nu(J) \rightarrow \visins{\policyViews}{\source}$}                                    	\label{alg:hideF:xi:begin}
      \State $\mu \defeq \origin{\bag}$
      \For{\textbf{each} $x \in \dom{\xi}$}                                                                         			 	\label{alg:hideF:x:begin}
          \If{$\xi(x) \neq *$}
                \For{\textbf{each} $y \in \frontier{\mu}$}
                      \If{$\tau(y) = \nu^{-1}(x)$, where ${\tau = \sub{\bag}}$}								\label{alg:hideF:x:violation}
                            \State \textbf{remove} $y$ from $\frontier{\mu}$
                      \EndIf
                \EndFor
          \EndIf
      \EndFor                                                                                                        						\label{alg:hideF:x:end}
      
      \If{$\mu \neq \origin{\bag}$}
        \State \textbf{add} $\mu$ to $\mathcal{R}$
      \EndIf
      
    \EndFor
    \State \textbf{choose} the best repair $\mu_r$ of $\mu$ from $\mathcal{R}$ based on $\pref$
    \State \textbf{return} $\mu_r$
    
\end{algorithmic}
\end{algorithm}

\begin{algorithm}[tb!]
\scriptsize
\caption{$\elimVars(\mu_1,\mu_2,\pref)$}\label{alg:elimVars}
\begin{algorithmic}[1]
    \State $\mathcal{R} := \emptyset$
    \If{$\exists$ one or more repeated variables in $\body(\mu_1)$}  
        \For{\textbf{each} ${\xi:\body(\mu_2) \rightarrow \body(\mu_1)}$
            mapping some ${x_1 \in \frontier{\mu_1}}$  
            \Statex \hspace{1.25cm} into some ${x_2 \not\in \frontier{\mu_2}}$}
            \State Let ${B \subseteq \body(\mu_1)}$, s.t. ${\xi(\body(\mu_2)) = B}$
            \State Let $V$ be the set of repeated variables from $B$
            \State Let $P$ be the set of positions from $B$, where all variables from $V$ occur
            \For{\textbf{each} non-empty ${S \subset P}$}
                \State ${\mu \defeq \mu_1}$ 
                \State \textbf{replace} the variables in positions $S$ of $\mu$ by fresh variables 
                \State \textbf{add} $\mu$ to $\mathcal{R}$
            \EndFor 
        \EndFor
    \EndIf
    \State \textbf{choose} the best repair $\mu_r$ of $\mu$ from $\mathcal{R}$ based on $\pref$
    \State \textbf{return} $\mu_r$
\end{algorithmic}
\end{algorithm}

\begin{example}
    We demonstrate Algorithm~\ref{alg:sphase} 
    over a simplified version of the running example, where $\Sigma'_{st} = \{\mu_e, \mu_c\}$.
    It is see that $\vchase{\Sigma'_{st}}{\source}$ will consist of the bags 
    $\{\bag_2,\bag_3,\bag_4,\bag_5\}$. We assume that ${n=\infty}$.
    During the first iteration of Algorithm~\ref{alg:sphase}, 
    the lowest depth bag for which there exists a homomorphism violation is $\bag_4$. 
    Since ${i < n}$, the algorithm tries to repair $\Sigma'_{st}$ by calling $\hideF$ and $\elimVars$
    with arguments (apart from $\policyViews$ and $\pref$) $\bag_4$ and ${\bag_2, \bag_4}$, respectively.  
    
    Algorithm~\ref{alg:hideF} first computes ${\nu = \{\lnu'_{\pid} \mapsto x_1, \lnu'_{\name} \mapsto x_2, \lnu_{\ethn} \mapsto x_3\}}$, 
    lines \ref{alg:hideF:nu:begin}--\ref{alg:hideF:nu:end}, and then computes all homomorphisms
    from $${\nu(J) = \{\pat(x_1,x_2,x_3,\ast), \hospn(x_1,\ast)\} }$$ 
    into the instance $\visins{\policyViews}{\source}$, 
    line~\ref{alg:hideF:xi:begin}. We can see that there exists only one such 
    homomorphism ${\xi=\{ x_1 \mapsto \lnu'_{\pid}, x_2 \mapsto \lnu'_{\name}, x_3 \mapsto \lnu_{\ethn}\}}$. 
    We have ${\origin{\bag_4} = \mu_e}$. The first two iterations of the loop in 
    lines~\ref{alg:hideF:x:begin}--\ref{alg:hideF:x:end} have no effect, since 
    despite that ${\xi(x_1) = \lnu'_{\pid}}$ and ${\xi(x_2) = \lnu'_{\name}}$, 
    the variables $\pid$ and $\name$ from $\mu_e$ that are mapped to $\lnu'_{\pid}$ and $\lnu'_{\name}$ via ${\sub{\bag_4} = h_4}$
    are not exported ones. During the last iteration, since ${\xi(x_3) = \lnu_{\ethn}}$,
    since ${h_4(\ethn) = \lnu_{\ethn}}$ and since $\ethn$ is an exported variable,    
    Algorithm~\ref{alg:hideF} removes variable $\ethn$ from the exported variables of $\mu_e$ and returns $\mu'_e$
    \begin{align}
        \pat(\pid,\name,\ethn,\post) \wedge \hospn(\pid,\dis)     &\leftrightarrow \ethreason(\dis)           \tag{$\mu'_e$} \nonumber
    \end{align}
    Algorithm~\ref{alg:sphase} then calls $\elimVars$. The function does not return any repair, 
    since there does not exist any variable repetition in the body of $\mu_e$. 
    Hence, Algorithm~\ref{alg:sphase} computes ${\Sigma_1 = \{\mu'_e,\mu_c\}}$ and proceeds in the next iteration.
    The instance $\visins{\Sigma_1}{\source}$ will consist of the following bags $\bag'_2$ and $\bag'_3$ with their corresponding homomorphisms shown below:
    \begin{align}
        \countyreason(\post,\dis) &\xrightarrow{{\langle \inverse{\mu_c}, h'_2 \rangle}} \pat(\lnu'_{\pid},\lnu'_{\name},\lnu_{\ethn},\ast), \hospn(\lnu'_{\pid},\ast)                                      &\nonumber\\
        \ethreason(\dis)    &\xrightarrow{{\langle \inverse{\mu_e}, h'_3 \rangle}} \pat(\lnu_{\pid},\lnu_{\name},\lnu'_{\ethn},\lnu_{\post}), \hospn(\lnu_{\pid},\ast)                                         &\nonumber
    \end{align}
    \begin{align}
        h'_2 &= \{\post \mapsto \ast, \dis \mapsto \ast \} \nonumber \\
        h'_3 &= \{\dis \mapsto \ast \} \nonumber 
    \end{align} 
    Algorithm~\ref{alg:sphase} terminates, since all bags are safe.
\end{example}

Note that when we reach the maximum number of iterations we do not apply $\elimVars$. 
This is due to the fact that $\elimVars$ might lead to unsafe unification of labeled 
nulls to $\ast$ that were not taking place before the modifying the s-t tgd through $\elimVars$. 
In contrast, $\hideF$ is a safe modification, since it does not lead to new unsafe unifications. 

\begin{restatable}{theorem}{thmsrepair}\label{theorem:srepair}
    For any partially-safe ${\mappingsOriginal = (\source,\target,\Sigma)}$, 
    any ${\mappingsOriginal_\viewSchema = (\source,\viewSchema,\policyViews)}$, 
    any preference function $\pref$ and ${n \geq 0}$, 
    Algorithm $\sphase$ returns a mapping ${\mappingsOriginal' = (\source,\target,\Sigma')}$ 
    that \emph{preserves the privacy of  
    ${\mappingsOriginal_\viewSchema}$ on all instances of $\source$}.
\end{restatable} 
\begin{proof}(Sketch)
	First note that since $\sphase$ takes as input a partially-safe mapping 
	${\mappingsOriginal = (\source,\target,\Sigma)}$, 
	it follows from Definition~\ref{definition:potential-safety} that 
	there exists a homomorphism from
      ${\chase{\inverse{\Sigma}}{\critical{\target}} \setminus \critical{\target}}$ 
      into $\visins{\policyViews}{\source}$. 
      Furthermore, from Proposition~\ref{proposition:potential-safety}, we know that 
      for each ${\mu \in \Sigma}$, there exists a homomorphism 
      from $\body(\mu)$ into $\visins{\policyViews}{\source}$ mapping each 
      ${x \in \frontier{\mu}}$ into the critical constant $\ast$. 
      Due to the above, 
      since the steps in lines~\ref{bagchase:EGDsChase:init}--\ref{bagchase:EGDsChase:create} 
      of Algorithm~\ref{alg:bagchase}
      do not introduce new labeled nulls
      and since $\sphase$ applies the procedure $\hideF$ to each unsafe bag $\bag$   
	in $B_n$, if there does not exist a bag ${\bag' \in B_n}$, such that  
	${\bag \prec \bag'}$, it follows that $\mappingsOriginal'$ preserves the privacy of  
      ${\mappingsOriginal_\viewSchema}$ on all instances of $\source$, if 
      $\hideF$ prevents dangerous unifications of labeled nulls with the critical constant 
      in line~\ref{bagchase:sigma_approx:end} of Algorithm~\ref{alg:bagchase}. 
      In particular, assume that we are in the $n$-th iteration of the steps in lines~ 
      \ref{alg:sphase:i:start}--\ref{alg:sphase:i:end} of Algorithm~\ref{alg:sphase}.
      Let ${\bag_{n}^0,\dots,\bag_{n}^M}$ be the unsafe bags in $B_n$.
      Assume also that for each ${1 \leq l \leq M}$, $\bag_{n}^l$, was derived due to some 
      active trigger $h^l$, for some derived egd ${\varepsilon^l \in \Sigma_{\approx}}$ in $I_j$,
      where ${j \geq 0}$, line~\ref{bagchase:EGDsChase:homom} of Algorithm~\ref{alg:bagchase}. 
      Let ${\mu^l = \origin{\varepsilon^l}}$, for each ${0 \leq l \leq M}$ and let $\mu_r^l$ be the repaired s-t tgd.
      Finally, let ${\bag_{n+1}^0,\dots,\bag_{n+1}^N}$ be the bags in $B_{n+1}$, line~\ref{alg:sphase:incremental}
      of Algorithm~\ref{alg:sphase}. 
      Based on the above, in order to show that Theorem~\ref{theorem:srepair} holds, 
      we need to show that (i) the number of bags in $B_{n+1}$ is $\leq$ the number of bags in 
      $B_{n}$ and that (ii) the s-t tgds in ${\left( \Sigma \setminus \bigcup \nolimits_{l=0}^{M} \mu^l \right) \cup \bigcup \nolimits_{l=0}^{M} \mu_r^l}$ are safe. 
      In order to show (i) and (ii), we consider the steps in Algorithm~\ref{alg:hideF}: for each ${1 \leq l \leq M}$,
      each exported variable $y$ occurring in $\mu^l$, which leads to an unsafe unification, line~\ref{alg:hideF:x:violation} of Algorithm~\ref{alg:hideF}, is turned into a non-exported variable. 
\end{proof}

By combining Proposition~\ref{proposition:frepair} and Theorem~\ref{theorem:srepair} we can prove the correctness of 
Algorithm~\ref{alg:repair}. Furthermore, if the preference function always prefers the repairs computed by $\hideF$
from the repairs computed by $\elimVars$, we can show the following: 
\begin{proposition}\label{proposition:completeness}
    For each mapping ${\mappingsOriginal = (\source,\target,\Sigma)}$, 
    each ${\mappingsOriginal_{\viewSchema} = (\source,\viewSchema,\policyViews)}$ and
    each preference function $\pref$ that always prefers the repairs computed by $\hideF$
    from the repairs computed by $\elimVars$, 
    Algorithm~\ref{alg:repair} returns a non-empty mapping that is safe w.r.t. $\mappingsOriginal_{\viewSchema}$, 
    if such a mapping exists. 
\end{proposition}
\begin{proof}(Sketch)
	From Algorithm~\ref{alg:frepair}, we can see that $\fphase$ always computes a non-empty partially-safe mapping, if such a mapping exists. 
	Note that a mapping, where no variable is exported and no repeated variables occur in the body of the s-t tgds is always partially-safe
	as long as, the predicates in the bodies of the s-t tgds are the same with the ones occurring in the policy views. 
	Please also note that such a mapping is always considered by $\fphase$.
	The above argument, along with the fact that a partially-safe mapping can be transformed 
	into a safe one by turning
	exported variables into non-exported ones by means of the function $\hideF$,
	show that Proposition~\ref{proposition:completeness} holds. 
\end{proof}


%% file: experiments.tex
\section{Experiments} \label{sec:experiments}

We investigate 
the efficiency of our repairing algorithm with the use of hard-coded preference function and with a preference function based on a learning approach.
The source code and the experimental scenarios are publicly available at
\url{https://github.com/ucomignani/MapRepair.git}.

\begin{table}[tb]
\centering
 \begin{tabular}{l|c|c|c}
        & min & max & step\\
 \hline
  \# s-t tgds per scenario ($n_{dep}$) & 100 & 300 & 50\\
  \# body atom per s-t tgds ($n_{atoms}$) & 1 & 3 (5) & $-$ \\
  \# exported variables per s-t tgds  ($n_{vars}$) & 5 & 8 & $-$ \\
 \hline
\end{tabular}
\caption{Properties of the generated iBench scenarios.}
\label{tab:iBench-conf}
\end{table}

We evaluated our algorithm using a set of 3,600 scenarios 
with each scenario consisting of a set of policy views and a set of s-t tgds. 
The source schemas and the policy views have been synthetically
generated using {\tt iBench}, the state-of-the-art data integration
benchmark \cite{arocena2015ibench}.
We considered relations of up to five attributes and we created GAV mappings 
using the {\tt iBench} configuration recommended by \cite{arocena2015ibench}. 
We generated policy views by applying the {\tt iBench} operators 
copy, merging, deletion of attributes and self-join ten times each.
The characteristics of the scenarios are summarized in 
Table~\ref{tab:iBench-conf}. In each scenario, we used a different number of s-t tgds $n_{dep}$, 
a different number of body atoms $n_{atoms}$ and
a different number of exported variables $n_{vars}$. 

We implemented our algorithm in {\tt Java} and we used the {\tt Weka} library \cite{eibe2016weka}
that provides an off-the-shelf implementation of the k-NN algorithm\footnote{\url{https://github.com/ucomignani/MapRepair/blob/master/learning.md}}.
We ran our experiments on a laptop with one 2.6GHz 
2-core processor, 16Gb of RAM, running Debian 9.

In the remainder, all data points have been computed as an average on five runs preceded by one discarded cold run. 

\subsection{Running time of $\repair$}\label{subsec:execution-time}
First, we study the impact of the number of s-t tgds 
and of the number of body atoms on the running time of $\repair$.
We adopt a fixed preference function that chooses the repair with the maximum number of exported variables,
while, in case of ties, it chooses the repair with the maximum number of joins.
We range the number of s-t tgds from $100$ to $300$ by steps of $50$
and the number of body atoms from three to five.
The results are shown in Figure~\eqref{fig-execTimes-low}.
Figure~\eqref{fig-execTimes-low} shows that the performance 
of our algorithm is pretty high; the median repairing time is less than $1.5$s,
while for the most complex scenario containing up to five body atoms per s-t tgd,   
the median running time is less than $8$s with $71$s being the maximum. 

\begin{figure*}
    \captionsetup[subfigure]{justification=centering}
    \centering
    \begin{subfigure}[b]{0.48\textwidth}
        \centering
        \includegraphics[width=\textwidth]{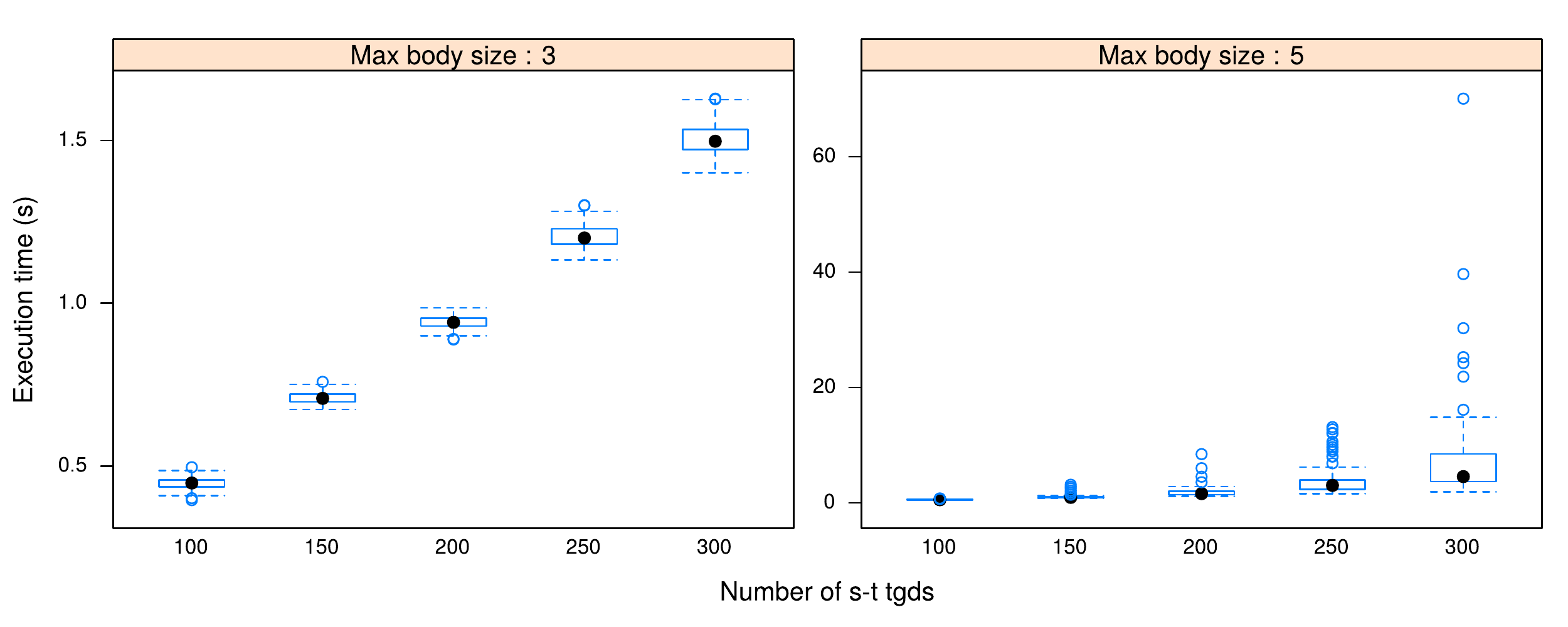}
	\caption{Repairing times.}
	\label{fig-execTimes-low}
    \end{subfigure}
    \quad
    \begin{subfigure}[b]{0.48\textwidth}
        \centering
        \includegraphics[width=\textwidth]{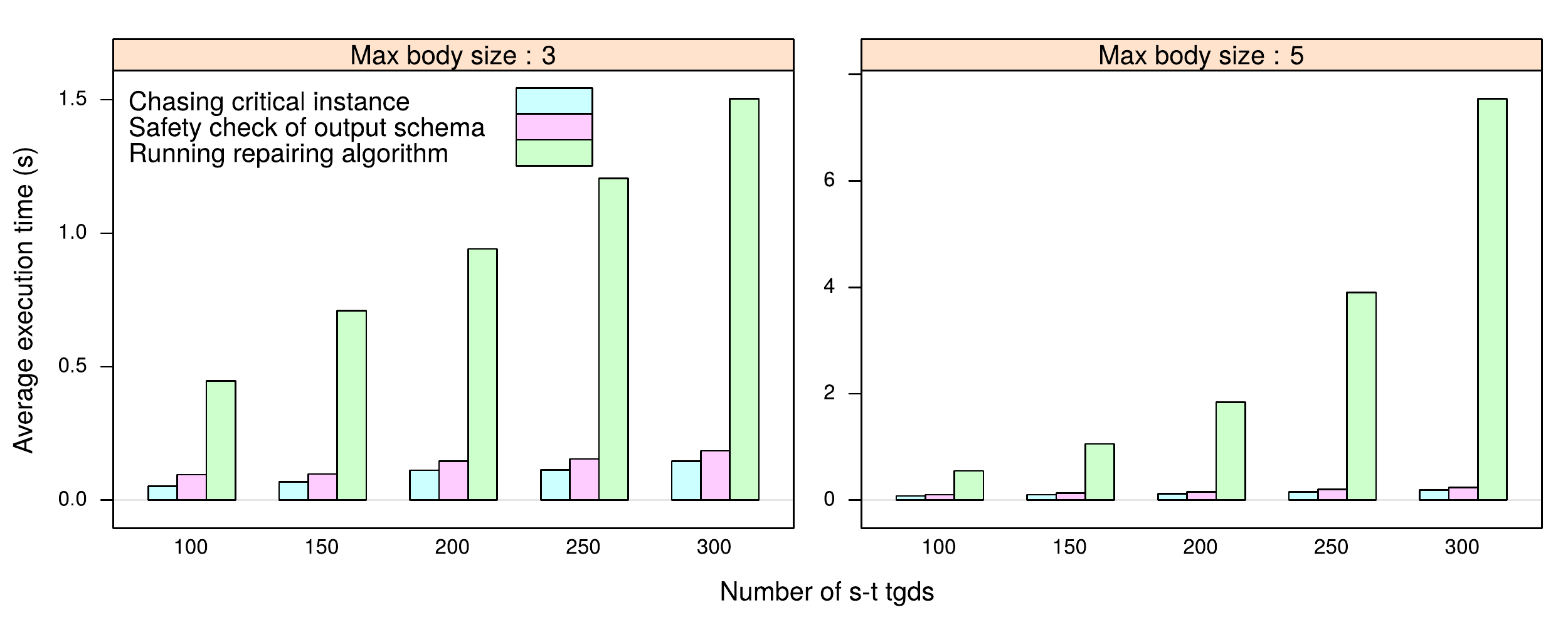}
	\caption{Time comparisons.}
	\label{fig-comparison}
    \end{subfigure}
    
    \begin{subfigure}[b]{0.48\textwidth}
	\centering
	\includegraphics[width=\textwidth]{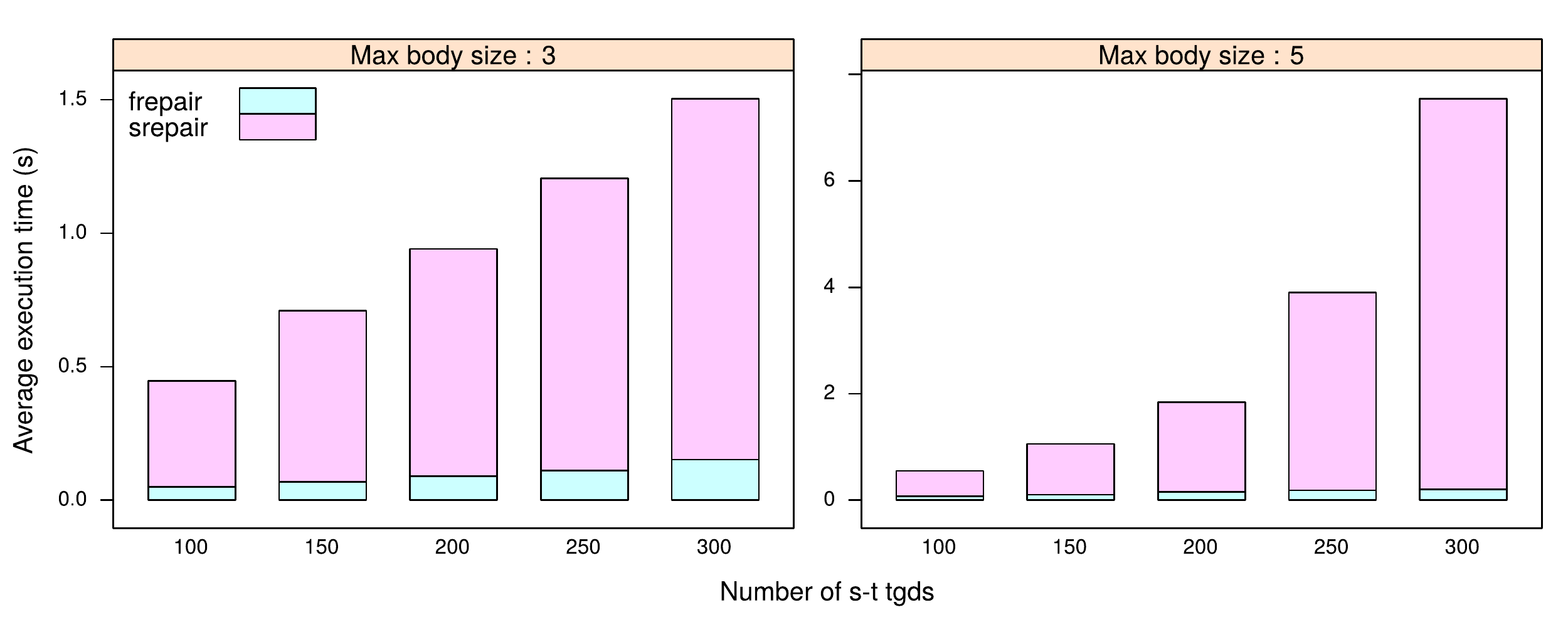}
	\caption{Time breakdown between $\fphase$ and $\sphase$.}
	\label{fig-breakdown}
    \end{subfigure}
    \quad
    \begin{subfigure}[b]{0.48\textwidth}
	\centering
	\includegraphics[width=0.6\textwidth]{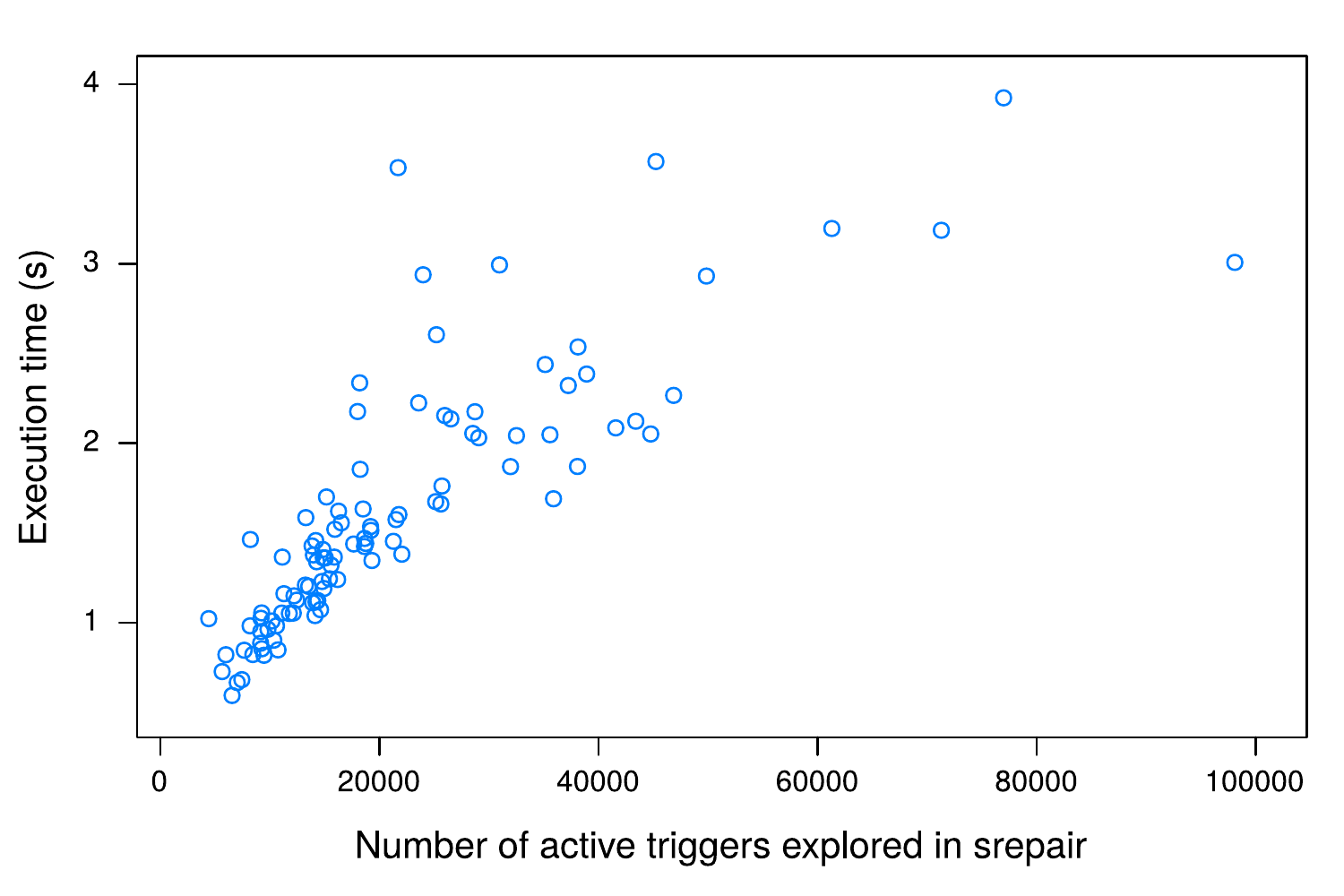}
	\caption{Running time of $\sphase$ over 100 s-t tgds.}
	\label{fig-activeTriggersExplo}
    \end{subfigure}    

    \caption{Summary of the performance-related experimental results.}\label{fig:expe1}
\end{figure*}

Figure~\eqref{fig-comparison} shows the time breakdown for 
$\repair$. The first column shows   
the average running time to run the visible chase over the input s-t mappings, 
the second one shows the average running time for checking the safety of the computed bags 
and the third one shows the average running time for repairing the s-t tgds.
The results show that the repairing time is $32$ times greater 
than tine to compute the visible chase and $40$ times greater 
than the time to check the safety of the chase bags for scenarios with 300 s-t tgds. 
In the simplest scenarios, these numbers are reduced to five and nine, respectively.
Overall, the absolute values of the rewriting times are kept low for these scenarios 
and gracefully scale while increasing the number of s-t tgds and the number of atoms in their bodies.  

\subsection{Time breakdown between $\fphase$ and $\sphase$}
Figure~\eqref{fig-breakdown} shows the average running time for $\fphase$ and $\sphase$
for the considered scenarios. We can see that $\sphase$ is the most time-consuming step of our algorithm.
We can also see that the running time of $\sphase$ increases more in comparison 
to the running time of $\fphase$ when increasing the number of the s-t tgds and the number of  
atoms in their bodies. This is due to overhead that is incurred during the incremental computation of the visual 
chase after repairing a s-t tgd (line~\ref{alg:sphase:incremental} of Algorithm~\ref{alg:sphase}).  
Figure~\eqref{fig-activeTriggersExplo} shows the correlation between 
the number of active triggers detected while incrementally computing the visual chase
and the running time of $\sphase$ for scenarios with 100 s-t tgds using the ANOVA method ({\it p-value} $< 2.2e^{-16}$).
Figure~\eqref{fig-activeTriggersExplo} shows that the most 
complex scenarios lead to the detection of more than $45,000$ active triggers.
Despite the high number of the detected active triggers, 
the running time of $\sphase$ is kept low thus validating its efficiency.

\subsubsection{Evaluating learning accuracy and efficiency}
We adopted the following steps in order to evaluate the performance of our learning approach. 
First, we defined the following two golden standard preference functions that we will try to learn: 
\begin{compactitem}
    \item $P_{max}$, which chooses the repair with the maximum number of exported variables 
    (i.e., the first repair if $\Delta_{FV} < 0$, else the other repair) and in case of ties, it chooses
    the repair with the maximum number of joins (i.e., the first repair if $\Delta_{J} < 0$, else the second repair).
    \item $P_{avg}$, which computes the average value $\Delta = \frac{\Delta_{FV} + \Delta_{J}}{2}$ 
    and chooses the first repair, if $\Delta < 0$; otherwise, it chooses the second repair.
\end{compactitem}
For both preference functions, we created a training set of $10,000$ measurements for the k-NN classifier by
running the repairing algorithm on fresh scenarios of 50 s-t tgds and five body atoms per s-t tgd.
For each input vector ${\langle \delta_{FV},\delta_{J} \rangle}$ whose repair we wanted to predict,
we computed the Euclidean distance between ${\langle \delta_{FV},\delta_{J} \rangle}$ and the vectors of the training set. 
We also set the value of parameter $k$ to $1$. This parameter controls the number of neighbors used to predict the output. 
Higher values of this parameter led to comparable predictions and are omitted for space reasons. 
Finally, we used the trained k-NN classifier as a preference function in $\sphase$, 
rerun the scenarios from Section~\ref{subsec:execution-time}
and compared the returned repairs with the ones returned when applying the 
golden standards $P_{max}$ and $P_{avg}$ as preference functions.

\emph{Learning $P_{max}$.}
Table~\eqref{fig:confucionMatrixPmax} 
shows the confusion matrix associated to learning $P_{max}$.
The confusion matrix outlines the choices undertaken during the iterations of the k-NN algorithm.
In our case, Table~\eqref{fig:confucionMatrixPmax} shows that $\mu_1$ has been selected 230 times, while $\mu_2$
has been chosen 395,680 times. 
We can thus see that $\mu_2$ is chosen in the vast majority of the cases.
Notice that $\mu_2$ is also the default value 
in cases where the preference function weights equally $\mu_1$ and $\mu_2$. 

Apart from the confusion matrix, we also measured the accuracy of learning the preference function, 
by weighing the closeness of the learned mapping to the golden standard mapping. 

We used the Matthews Correlation Coefficient metric (MCC) \cite{baldi2000assessing} to 
compare the repairs returned by the trained k-NN classifier and the ones returned when applied  
$P_{max}$. This is a classical measure that 
allows to evaluate the quality of ML classifiers when ranking is computed 
between two possible values (in our case, the choice between $\mu_1$ and $\mu_2$).
This measure is calculated using the following:
\begin{itemize}
 \item $N_{1,1}$ the number of predictions of $\mu_1$ when $\mu_1$ is expected
 \item $N_{2,2}$ the number of predictions of $\mu_2$ when $\mu_2$ is expected
 \item $N_{1,2}$ the number of predictions of $\mu_1$ when $\mu_2$ is expected
 \item $N_{2,1}$ the number of predictions of $\mu_2$ when $\mu_1$ is expected
\end{itemize}
$$MCC =  \frac{N_{1,1} \times N_{2,2} - N_{1,2} \times N_{2,1}}
	      { \sqrt{ 
	      (N_{1,1} + N_{1,2})
	      (N_{1,1} + N_{2,1})
	      (N_{2,2} + N_{1,2})
	      (N_{2,2} + N_{2,1})
	      } } $$
The results of $MCC$ range from $-1$ for the cases where the model perfectly predicts the inverse of the expected values, 
to $1$ for the cases where the model predicts the expected values. 
The value $MCC = 0$ means that there is no correlation between the predicted value and the expected one.
By applying MCC to the learning of $P_{max}$, we observed that 
the data are clearly discriminated, thus leading to a perfect fit of our prediction in this case ($MCC=1$).

\emph{Learning $P_{avg}$.}
Table~\eqref{fig:confucionMatrixPavg} shows the confusion matrix 
associated to learning $P_{avg}$.
We can see that the predictions are less accurate in this case.
The data is not as clearly discriminated as before, leading to a fairly negligible error rate ($<0.02\%$). 
This error is still acceptable for the learning, since only $<0.02\%$ of the predictions are erroneous.
This is corroborated by a $MCC$ value equal to $0.93$, 
thus leading to a still acceptable fit of our preference function also in the case of $P_{avg}$.

\begin{figure}
    \centering
    \includegraphics[width=1\columnwidth]{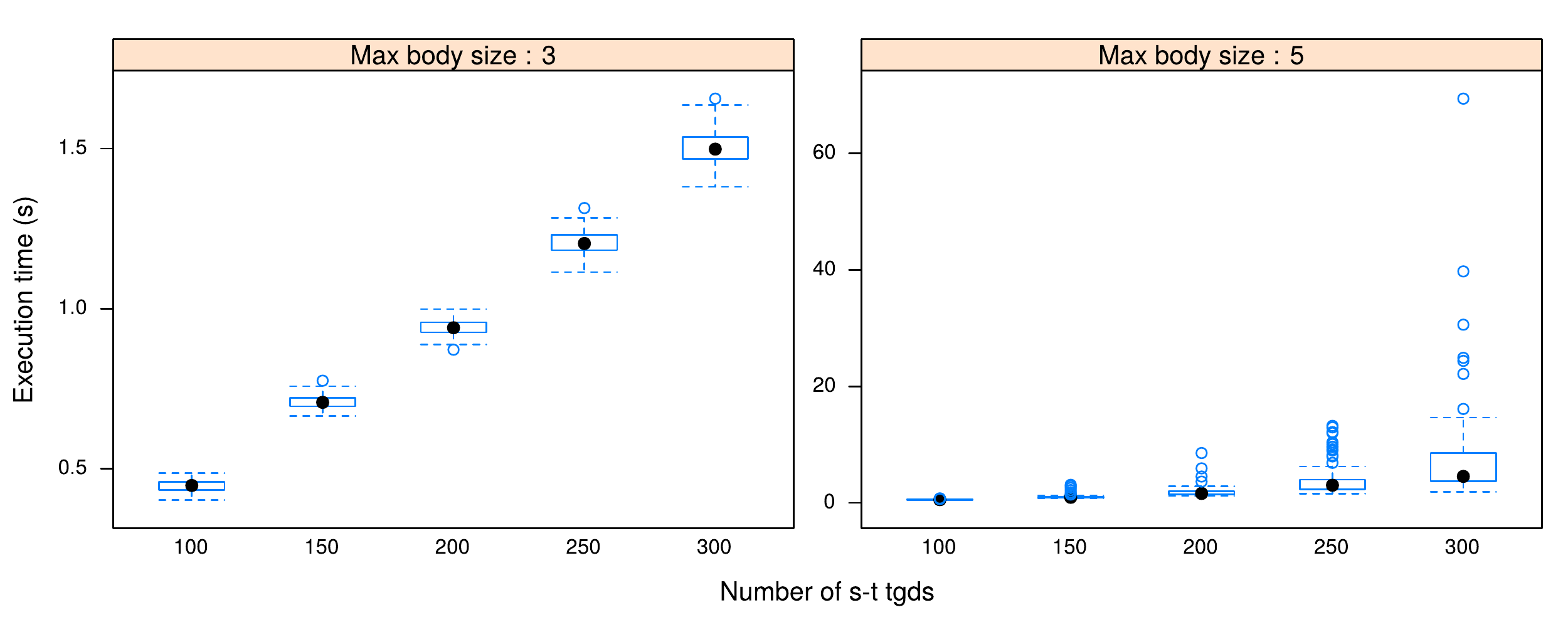}
    \caption{Repairing times with ML.}\label{fig:expeML}
\end{figure}
\subsubsection{Running time of $\repair$ with a learned preference function}\label{sec:learningExecTimes}

In the last experiment, we want to measure the impact of learning on the performance of our algorithm. 
To this end, we compare the running time of 
$\repair$ when adopting a hard-coded preference function (as in the results reported in Figure~\ref{fig:expe1}) 
and when adopting a learned preference function. Figure \ref{fig:expeML} shows the running times 
for the same scenarios used in Figure~\ref{fig:expe1}. We can easily observe that the runtimes are rather similar with and without learning and the difference amounts to a few milliseconds. This further corroborates the utility of learning the preference function and shows that the learning is robust and does not deteriorate the performances of our algorithm. 

\begin{table}
    \centering
    \begin{subfigure}[b]{0.48\columnwidth}
    \centering
      \adjustbox{height=0.9cm ,keepaspectratio}{
	  \begin{tabular}{c|cc}
	  \toprule
	   & \multicolumn{2}{c}{golden standard}\\
	  prediction & $\mu_1$ & $\mu_2$\\
	  \midrule
	    $\mu_1$ & 230 & 0 \\
	    $\mu_2$ & 0 & 395680 \\
	  \bottomrule
	  \end{tabular}
     }
     \caption{$P_{max}$ confusion matrix.}
     \label{fig:confucionMatrixPmax}
    \end{subfigure}
    ~
    \begin{subfigure}[b]{0.48\columnwidth}
    \centering
    \adjustbox{height=0.9cm ,keepaspectratio}{
	  \begin{tabular}{c|cc}
	  \toprule
		     & \multicolumn{2}{c}{golden standard}\\
	  prediction & $\mu_1$ & $\mu_2$\\
	  \midrule
	    $\mu_1$ & 290 & 1 \\
	    $\mu_2$ & 42 & 395577 \\
	  \bottomrule
	  \end{tabular}
    }
    \caption{$P_{avg}$ confusion matrix.}
    \label{fig:confucionMatrixPavg}
    \end{subfigure}

    \caption{Confusion matrix for the golden standards.
    }\label{fig:confucionMatrix}
\end{table}

%% file: conclusion.tex
\section{Conclusion} \label{sec:conclusion}

We have addressed the problem of safety checking w.r.t. a set of policy
views in a data exchange scenario. 
We have also proposed efficient repairing
algorithms that sanitize the mappings w.r.t. the policy views,
in cases where the former leak sensitive information. 
Our approach is inherently data-independent and leads to obtaining 
rewritings of the mappings guaranteeing privacy preservation at a
schema level. As such, our approach is orthogonal to several data-dependent
privacy-preservation methods, that can be used on the companion source and
target instances to further corroborate the privacy guarantees. 
We envision several extensions of our work, such as the study of general
GLAV mappings and the interplay between data-independent and
data-dependent privacy methods.